\documentclass[12pt]{article}

\usepackage{amsthm}
\usepackage{amsmath}
\usepackage{amsfonts}
\usepackage{bbm}
\usepackage{thmtools}
\usepackage{thm-restate}
\usepackage{physics}
\usepackage{geometry}
\usepackage[doublespacing]{setspace}
\usepackage[hidelinks]{hyperref}
\usepackage[authoryear]{natbib}

\declaretheorem[name=Corollary]{coro}
\declaretheorem[name=Assumption]{assn}
\declaretheorem[name=Definition]{defn}

\begin{document}

\title{International cooperation and competition in orbit-use management}

\author{Aditya Jain\\ \small Becker Friedman Institute for Economics
        \and 
        \href{https://akhilrao.github.io/pages/about.html}{Akhil Rao}\thanks{Department of Economics, Warner Hall, Middlebury College, 05753; akhilr@middlebury.edu} \\ \small Middlebury College}
\date{This version: September 10, 2022}

\maketitle

\begin{abstract}
 Orbit-use management efforts can be structured as binding national regulatory policies or as self-enforcing international treaties. New treaties to control space debris growth appear unlikely in the near future. Spacefaring nations can pursue national regulatory policies, though regulatory competition and open access to orbit make their effectiveness unclear. We develop a game-theoretic model of national regulatory policies and self-enforcing international treaties for orbit-use management in the face of open access, regulatory competition, and catastrophe. While open access limits the effectiveness of national policies, market-access control ensures the policies can improve environmental quality. A large enough stock of legacy debris ensures existence of a global regulatory equilibrium where all nations choose to levy environmental regulations on all satellites. The global regulatory equilibrium supports a self-enforcing treaty to avert catastrophe by making it costlier to leave the treaty and free ride.
\end{abstract}

JEL codes: F53, H77, Q53, Q54, Q58 \par
Keywords: open-access commons, satellites, space debris, international treaty, regulatory competition

\newpage
\section{Introduction}
%% Hook: Attract the reader’s interest by telling them that this paper relates to something interesting. What makes a topic interesting? This should be one clean paragraph.

Orbital space is a critical resource for global economic development. Indeed, a recent assessment found that satellite services make significant contributions to 13 of 17 UN Sustainable Development Goals \citep{unoosa2018european}. Yet orbital space is a global commons under an open access regime. International law limits the degree to which orbit users can exclude each other from valuable orbital slots, while the laws of physics prevent multiple objects from occupying the same slots. Objects attempting to occupy the same slots will collide, creating long-lived and fast-moving debris which further increases collision risk. As debris proliferates, collisions between debris objects may cross a tipping point beyond which debris growth becomes self-sustaining, a condition known as ``Kessler Syndrome''. While the relevant physical dynamics operate over years and decades, crossing the tipping point would be a global catastrophe which would severely limit human use of outer space. \\

%% Question: Tell the reader what this paper actually does. Think of this as the point in a trial where having detailed the crime, you now identify a perpetrator and promise to provide a persuasive case. The reader should have an idea of a clean research question that will have a more or less satisfactory answer by the end of the paper. The question may take two paragraphs. At the end of the first (2nd paragraph of the paper) or possibly beginning of the second (3rd paragraph overall) you should have the “This paper addresses the question” sentence. All of this is the "standard" recommendation for research (rather than survey) articles, so perhaps we do something different, like "This paper surveys the existing economic literature on space debris and orbit use management. We also offer a unifying simple model through which different mechanisms and policy proposals can be discussed."

How can orbital space be managed through binding policies at the national or international levels? Are there trade-offs between national and international approaches? To what extent does the risk of Kessler Syndrome support or hinder these policy approaches? And how does open access to orbital space---a fundamental constraint imposed by the Outer Space Treaty (OST), a foundational pillar of international space law---affect the prospects for national or international orbital-use management? This paper addresses these questions through a game-theoretic model of rational orbit use. While components of these questions have been studied separately, there is relatively little work exploring the interactions between national and international orbital-use management approaches given the relevant physical and institutional constraints. Understanding these interactions is critical to designing effective orbital-use management policies and predicting policy outcomes.\footnote{Specifically, Articles I, II, and VI of the OST specify free access to outer space (I), non-appropriation of outer space (II), and national supervision of space activities by non-governmental entities (VI). Articles I and II render orbital space an open-access commons, while Article VI enables national satellite regulatory policies.} \\

%% Antecedents: Identify the prior work that is critical for understanding the contribution this paper will make. The key mistake to avoid here are discussing papers that are not essential parts of the intellectual narrative leading up to your own paper.

Our paper connects to three literatures in economics and environmental management. First, we extend the literature on orbital-use management. Prior work in this area has established several stylized facts: the inefficiency of decentralized orbit use \citep{adilov2015economic, rouillon2020physico}; the importance of tipping points and decentralized rational actors' inability to avoid crossing them \citep{lewis2020understanding, rao2022cost}; and the potential for satellite tax-like policies to efficiently limit collision risk and debris accumulation while increasing the value generated by satellites \citep{rao2018economic,rao2020orbital, beal2020taxing}. However, this literature has not yet considered how management policies can be structured given international open access and national competition for satellite services, and with the exception of \cite{singer2011international} it has largely ignored the potential of debris abatement treaties. \\

Second, we contribute to the broader literature about environmental federalism and decentralized commons management. Prior work in this area has established conditions under which regulatory competition between jurisdictions can effectively manage environmental quality \citep{oates1999essay, millimet2013environmental}, particularly in the face of limited regulatory control over spatially-mobile resources and public bads \citep{costello2015partial, costello2017private}. This literature has largely focused on cases where resource users and public bads are spatially separated and there are no catastrophic thresholds, and has not studied the interactions between national policies and international treaties. Orbital space is novel in that ``partial enclosure of the commons'' (i.e. limiting access to a specific physical space to prevent resource collapse) is physically and legally infeasible.\footnote{The notion of commons enclosure relevant to orbital space is better described as ``limited enclosure''. In the fisheries context of \cite{costello2015partial}, ``limited enclosure'' would involve one nation preventing boats under their flag from fishing in the high seas. Such attempts are unlikely to be as effective in preserving marine populations as limiting access to the high seas by all nations, since boats could fly flags of convenience to circumvent the regulations and sell their catch at the world market price. However, a key difference between orbit use and many other terrestrial natural resources is that satellite services typically can't be resold like fish or physical goods, giving nations an additional locus of control unavailable in most other resource contexts.}. \\

Third, we contribute to the literature on self-enforcing international environmental treaties. Prior work in this area has considered how uncertainty over a catastrophe threshold affects treaty formation \citep{barrett2013climate, barrett2014sensitivity} and even how a space debris abatement treaty could be structured \citep{singer2011international}. However, the literature has not yet considered pollutant abatement treaties in the face of catastrophe, national policies, and open access to the polluted commons by profit-maximizing firms. \\

%% Value-Added: Describe approximately 3 contributions this paper will make relative to the antecedents. This paragraph might be the most important one for convincing referees not to reject your paper. A large difference between it and the earlier “question” paragraph is that the contributions should make sense only in light of prior work whereas the basic research question of the paper should be understandable simply in terms of knowing the topic (from the hook paragraph). John suggests that “Antecedents” and “Value-added” may be intertwined. They may also take up to 3 paragraphs.

In this paper we build a game-theoretic model of orbit use by commercial satellite operators and spacefaring nations. Our model features profit-maximizing firms competing to provide satellite services under open access, debris accumulation and collision damages with a catastrophe threshold, national efforts to attract satellite services and regulate them, and national cooperation to abate debris and avert catastrophe. Our model enables us to identify three key features about orbital-use management through national policies and international treaties. First, despite satellite operators migrating to nations with less-stringent policies, national policies can improve all nations' welfare from satellite services and orbital environmental quality. Open access therefore will not render national regulatory policies fruitless. Second, if the legacy debris stock is large enough, a global regulatory equilibrium can exist where all nations impose regulations on their own and each others' satellite fleets without transfers of surplus between nations. Third, by altering the marginal benefits of debris abatement, national regulatory policies can strengthen a self-enforcing treaty to prevent Kessler Syndrome. \\

Our results bring together several threads from the existing literature on commons governance, both in general and in the space debris context. \cite{weeden2012taking} draw on Ostrom's principles for commons governance to suggest that a nested system of governance may be a successful approach to orbital-use management. \cite{johnson2012application} express reservations about the potential for actually implementing such an overlapping system of orbital governance. \cite{wouters2016space} argue that a new debris abatement treaty is unlikely in the near future and instead focus on policy actions European nations can take to manage orbit use. Similarly, \cite{migaud2020protecting} and \cite{gilbert2021major} focus on actions the US can take to manage orbit use using only existing national policy authorities. \cite{munoz2018regulating} points to OST interpretations which support treating debris as salvage to facilitate debris removal. We offer guidance on implementing these ideas while accounting for the responses of satellite operators and other nations, as well as indications of how they may support further improvements in orbital-use management. Our results are also informative about other global common-pool resources subject to global open access, environmental tipping points, and a mix of national and international management practices. These include the atmosphere, the North Pole, and the lunar surface. \\

%% Road-map: Outline the organization of the paper. Avoid writing an outline so generic that it could apply to any paper (“the next section is the middle of the paper and then we have the end”). Instead customize the road map to the project and possibly mention pivotal “landmarks” (problems, solutions, results…) that will be seen along the way. But keep this short because many readers will now be eager to get to the heart of the paper.

This paper is organized as follows. In section \ref{sec:model} we describe the model and present key results. We begin with interactions between firms and the physical environment, then consider national incentives and regulations, and finally explain the international treaty for net debris abatement. Throughout we provide results illustrating properties of each component of the system. We discuss the significance of our results, their connections to existing literatures, and their relation to unmodeled national security issues in section \ref{sec:discussion}. We conclude in section \ref{sec:conclusion}. All proofs are shown in the Appendix.

\section{Model and results}
\label{sec:model}

We focus on long-run orbit use by competitive satellite operators.\footnote{``Long-run'' in this context should not be interpreted as ``steady state'', as we allow for tipping points to be crossed and damages to accrue. Instead, it should be interpreted as a period sufficiently far in the future that we can consider the present value of accumulated streams of benefits and costs. Similar approaches are taken in \cite{weitzman1998far} and \cite{delong2012fiscal}, where catastrophe, long-run damages or output gaps, and policies to mitigate them feature prominently.} As such, we abstract from the details of orbital-use dynamics and focus on two properties of orbit use established in \cite{somma2019sensitivity}, \cite{lewis2020understanding}, and \cite{rao2022cost}. First, the long-run collision risk is a function of the long-run launch rate, with higher launch rates implying greater collision risk. We therefore treat collision risk to satellites as an increasing function of the debris stock and the debris stock as an increasing function of the number of satellites maintained in orbit. Larger satellite stocks imply greater collision-related damages. Second, there is a threshold level of debris beyond which debris growth becomes self-sustaining. Such debris growth can eventually render a region of orbital space unusable on the order of years or decades, and crossing the threshold will not immediately stop rational economic agents from launching satellites. We therefore model the long-run damages from crossing the threshold as a lumpsum value incurred immediately, reflecting the present discounted value of a stream of damages which slowly increase over time.\footnote{Following \cite{weitzman1998far}, we assume these streams are discounted at the lowest possible rate.} \\

Operators provide a homogeneous output, e.g. telecommunications, and face a constant exogenous price in each nation in which they operate.\footnote{While telecommunications products in reality are differentiated by factors like bandwidth and coverage, these features are unlikely to change our key messages regarding open access, regulations, and treaties.} The constant price assumption reflects the existence of a terrestrial substitute good and a competitive market, both limiting satellite operators' ability to set prices. These prices vary across nations to reflect their different market sizes and attractiveness to satellite operators. Operators can use the same satellite to provide services in multiple nations. The lack of excludability over orbital slots creates an open-access problem, wherein operators launch satellites until the marginal operator receives zero profits \citep{adilov2015economic, rouillon2020physico, rao2020orbital}.\footnote{Paths in low-Earth orbit are dynamic objects in 3-dimensional space, requiring 6 parameters---three positions and three velocities---to completely specify. Despite this, they can be characterized as ``slots'' in a higher-dimensional space through appropriate coordinate transformations \citep{arnas2021definition}.} In addition to prices for satellite services, nations are also distinguished by heterogeneous costs of launching and operating satellites from their territories (e.g. due to different launch technology or labor force availabilities). \\

Nations seek to regulate satellite operators in order to maximize the net benefits they receive from satellite services. These regulations are imposed as conditions for market access, i.e. the right to transmit and receive signals (more generally, sell and purchase products, e.g. satellite imagery) to and from customers located within a nation's jurisdiction.  We model binding national regulatory policies as ``virtual satellite taxes'', reflecting that any policy which alters firms' behavior does so by imposing costs on the firms. These taxes can therefore include policies such as regulatory barriers (e.g. compliance costs due to arms control or environmental regulations) or explicit taxes which generate revenues (e.g. orbital-use fees as described in \cite{rao2020orbital}).\footnote{While policies regulating orbiting satellites and satellite launches are equivalent in long-run models, their dynamics vary considerably, with satellite-focused policies being more efficient than launch-focused policies \citep{rao2018economic}. We use the ``satellite tax'' terminology to avoid short-run dynamic inefficiency issues.} We assume these taxes are ``bilateral'' between nations, i.e. levied by one nation on all satellites operated under another nation's jurisdiction, rather than differentiated by individual operators.\footnote{Operator-specific taxes would add more complexity without much additional insight regarding the questions considered here.} Nations receive two types of benefits from implementing regulatory policies. First, by reducing the total stock of satellites in orbit, satellite operators incur fewer collision-related satellite losses. Nations then receive more satellite services as existing satellites are destroyed less frequently. Second, to the extent that these taxes generate revenues, nations can use those revenues to provide public goods or reduce distortionary taxes. We assume that the taxes generate no revenues so as to keep the focus solely on benefits received through environmental management. \\

Finally, nations also negotiate a treaty for net debris abatement, and choose whether or not to be party to the treaty. We depart from the setting considered in \cite{singer2011international} by microfounding the national benefits from debris abatement, incorporating open-access behavior by satellite operators, and introducing the tipping point for Kessler Syndrome. While we derive the national benefits of debris abatement from the national net market benefits from satellite services, it is straightforward to include additional terms reflecting non-market economic benefits from satellite services.\footnote{Incorporating strategic interactions between nations, e.g. due to national security or dual-use concerns, is less straightforward. Such interactions introduce technical challenges relating to interdependent preferences. We briefly discuss this issue in section \ref{sec:discussion} but leave a full modeling treatment for future research.} 

\subsection{Satellites, debris, and catastrophe}

Identical satellites can provide services to any of the $N$ nations (markets) in the world. There are $n_s \leq N$ nations from which satellites are operated. We refer to these as spacefaring nations (even if they do not possess their own launch capabilities). The satellite sector in spacefaring nation $i$ controls $S_i$ satellites. Launching and maintaining these satellites costs $m_i S_i^2$, where $m_i$ is a sector efficiency parameter reflecting the cost-minimizing mix of launch services, capital, and labor supply available to satellite operators based in nation $i$.\footnote{Firms may procure inputs internationally, in which case $m_i$ will also reflect exchange rates and relevant international regulations facing the satellite sector in nation $i$.} The price of a unit of satellite service in nation $j$ is a constant $p_j$, reflecting the size of the market for satellite services in nation $j$.\footnote{Constant prices approximate the case for satellite services with terrestrial substitutes, e.g. broadband.} Subscripted dots indicate arrays over the omitted index, e.g. $p_{\cdot} = (p_1, \dots, p_N)$, while negative indices indicate sums over all entries except the index shown, e.g. $S_{-i} = \sum_{i' \neq i}S_{i'}$. \\

Each active satellite produces $d$ units of debris over its lifecycle. The long-run stock of debris in orbit, $D$, is determined by the long-run total satellite stock, $S = \sum_{i=1}^{n_s} S_i$, the ``legacy'' stock of debris unrelated to ongoing orbit use by active satellites $D_0$, and net debris abatement by all nations $Q$, which is defined relative to the benchmark level of debris due to ongoing orbit use and the legacy debris stock, $dS + D_0$. The long-run debris stock is shown in equation \ref{eqn:debris_law}:
\begin{equation}
    D = dS + D_0 - Q.
    \label{eqn:debris_law}
\end{equation}

Debris abatement $Q$ contains two components: reduction of debris produced by satellites and debris removal activities. It is sometimes useful when discussing debris abatement (particularly in section \ref{sec:model_treaties}) to express the Kessler Syndrome threshold $\bar{D}$ in terms of the debris abatement required to avert it, $\bar{Q}$. The catastrophe-averting level of abatement, $\bar{Q}$, can be derived from equation \ref{eqn:debris_law} as follows:
\begin{align}
    D > \bar{D} \implies& dS + D_0 - Q > \bar{D} \nonumber \\
    \implies& Q < dS + D_0 - \bar{D} \nonumber \\
    \implies& \bar{Q} = dS + D_0 - \bar{D}.
\end{align}

Debris damages active satellites. The probability a satellite survives in orbit given debris stock $D$ is shown in equation \ref{eqn:sat_destrn}:
\begin{align}
    \label{eqn:sat_destrn}
    Pr(\text{survive}) &= 1 - kD,\\
    k : Pr(\text{survive}) &\in [0,1].
    \label{eqn:sat_destrn_k_restriction}
\end{align}

$k$ is the probability of collision per unit of debris, and condition \ref{eqn:sat_destrn_k_restriction} is a physical restriction imposed by this interpretation. If debris levels exceed a threshold $\bar{D}$, self-sustaining debris growth occurs. This is a slow-moving catastrophe: over many years, collisions between debris generate more debris, and eventually  that region of orbital space is rendered unusable. Satellite sectors incur common long-run damages of $X$ from crossing the threshold. There may be years or decades after the threshold is crossed during which satellite sectors continue to accrue profits before collision risk becomes large enough to make satellite operation unprofitable.\footnote{Estimates of damages to modern economies from the loss of satellite services suggest the damages parameter from Kessler Syndrome in valuable orbits can be large \citep{bradley2009space, schaub2015cost, o2019economic}, though many relevant structural parameters (e.g. substitution possibilities, indirect and induced demand effects) are uncertain and may reduce the damages \citep{van2021world, highfill2022estimating}. Given the potential for large damages and ``tail-fattening'' effects of structural parameter uncertainty \citep{weitzman2009modeling}, we ignore idiosyncratic national variation in the damages parameter.} Since our focus is on long-run outcomes we also abstract from the dynamics of debris growth and focus on whether or not the threshold is crossed. Following \citet{rao2022cost}, we assume satellite operators do not internalize the costs of crossing the self-sustaining debris growth threshold.  \\

Each nation receiving satellite services can impose regulatory burdens on satellite sectors. These include radio spectrum licensing procedures, explicit Pigouvian taxes on debris creation, satellite design requirements, and any other costly condition for market access. These regulatory burdens are satellite taxes, $\tau_{ij}$, imposed by market $j$ on sector $i$.\footnote{\cite{rao2018economic} shows that the optimal instrument in a dynamic context where objects can be deorbited and both active satellites and debris can induce collisions is a ``stock control'': a tax levied on all objects in orbit. \cite{rao2020orbital} describe these controls as ``orbital-use fees'', since they are in effect fees for using orbital space. \cite{guyot2021designing} show that orbital-use fees can be decomposed into separate streams of fees levied on launch debris and satellites at different stages of their lifecycle. In our setting, where we abstract from dynamics, make the debris stock linearly related to the satellite stock, and only allow debris to cause collisions, orbital-use fees are equivalent to satellite taxes.} Satellite taxes reduce service delivery from the taxed sector to the taxing market, so the quantity of satellite services received in market $j$ is $\sum_{i}(1-\tau_{ij})S_{i}$. Denial of satellite sector $i$'s access to market $j$ implies $\tau_{ij} = 1$. 
The net profits accrued by sector $i$, $Y_i$, are shown in equation \ref{eqn:satellite_profits}.
\begin{align}
    Y_i &= (1 - kD)\sum_{j=1}^N p_j (1 - \tau_{ij})  S_i - m_i S_i^2.
    \label{eqn:satellite_profits}
\end{align}

The inclusion of collision risk in front of the satellite stock in equation \ref{eqn:satellite_profits} reflects the long-run character of the model: of the $S_i$ satellites launched (reflected in cost term $m_i S_i^2$), a fraction $1 - kD$ will be destroyed in collisions and therefore unable to deliver services.\footnote{We focus on the costs of destructive collisions in this paper, but collision-avoidance maneuvers create another type of cost with similar effects. We abstract from these costs since they would introduce additional complexity without changing our results.} We use this long-run property to abstract from the dynamics of launch and destruction. The size of satellite sector $S_i$ is determined by an open-access condition, shown in equation \ref{eqn:open_access}:
\begin{equation}
    S_i^\star(S_{-i},\tau_{i\cdot}, Q) : Y_i = 0.
    \label{eqn:open_access}
\end{equation}

Equation \ref{eqn:open_access} reflects the fact that satellite operators in each nation are unable to secure exclusive property rights to orbital space, and therefore launch satellites until there are no further profits to doing so. Solving equation \ref{eqn:open_access} yields an expression for the size of satellite sector $i$ as a function of the satellite taxes $\tau_{i \cdot}$, the total size of all other satellite sectors $S_{-i}$, and the global level of debris abatement $Q$.\footnote{Since equation \ref{eqn:open_access} is quadratic in $S_i$ with no intercept, there are two roots: one where $S_i = 0$ and one where $S_i > 0$. We ignore the root where $S_i = 0$.}  The expression is shown in equation \ref{eqn:sat_bestresponse}:
\begin{equation}
    S_i^\star(S_{-i},\tau_{i\cdot}, Q) = \frac{\sum_j p_j (1 - \tau_{ij}) (1 - k(dS_{-i} + D_0 - Q))}{kd\sum_j p_j (1 - \tau_{ij}) + m_i}.
    \label{eqn:sat_bestresponse}
\end{equation}

The congestibility of orbital space suggests satellite sectors will compete with each other for physical space in orbit, i.e. national satellite fleets will be strategic substitutes. Differentiating equation \ref{eqn:sat_bestresponse} for arbitrary sectors $i$ and $j \neq i$ reveals this to be the case:
\begin{equation}
    \pdv{S^\star_i}{S_j} = - \frac{k d p_j (1 - \tau_{ij}) }{kd\sum_j p_j (1 - \tau_{ij}) + m_i} < 0.
    \label{eqn:sat_bestresponse_static}
\end{equation}

This presents a challenge to national regulatory approaches, as some of the fleet size reductions achieved by one sector will be partially offset by fleet size increases from others. This does not imply that reductions by sector $j$ have no effect on the total fleet size---as long as sector $i$ faces positive launch costs they will not fully offset $j$'s reduction ($m_i>0 \implies \pdv{S^\star_i}{S_j} > -1$). \\

A \textit{global open-access equilibrium}, formalized in Definition \ref{def:open_access}, is a distribution of satellites such that every satellite sector earns zero profits---a fixed point of all $i$ best-response functions in equation \ref{eqn:sat_bestresponse}.

\begin{defn}(Global open-access equilibrium)
\label{def:open_access}
    A global open-access equilibrium is a distribution of satellites $\{S_i^\star\}$ such that, given a debris abatement level $Q$, a legacy debris stock $D_0$, distribution of satellite taxes $\{\tau_{ij}\}$, market sizes $\{p_j\}$, sector efficiencies $\{m_i\}$, and technologies $(k,d)$, equation \ref{eqn:open_access} holds for all $i$.
\end{defn}

Lemma \ref{lemm:open_access_eqm} establishes existence and a characterization of the global open-access equilibrium with an arbitrary number of nations. Lemma \ref{lemm:reduction} establishes that as long as the equilibrium to the $n_s$-player game exists, there exists a $2$-player game such that the equilibrium satellite stock for sector $i$ is identical to $S^*_i$ in the $n_s$-player game, with the satellite fleet of the remaining nations being identical to $S^*_j$. This reduction allows us to use a $2$-player game to establish properties of the $n_s$-player game. Intuitively, the reduction is possible because any firm in nation $i$ only needs to know how many satellites are launched from other firms in nation $i$ (with whom it competes for satellite launch and operations resources as well as orbital space) and those launched from all other sectors (with whom it competes for orbital space).\footnote{That is, Lemma \ref{lemm:reduction} establishes that economic competition for orbital space has a form of ``diffuse reciprocity'' \citep{keohane1986reciprocity}.}\\

\begin{restatable}[Global open-access equilibrium]{lemm}{oaEqmExist}
\label{lemm:open_access_eqm}
Let $S$ be the $n_s \times 1$ vector of satellite sector sizes and $I$ be the $n_s \times n_s$ identity matrix. Let $A$ be the $n_s \times 1$ vector with typical elements
\begin{equation}
    A_i = \frac{\sum_j p_j (1 - \tau_{ij})(1 + k(Q-D_0))}{kd \sum_j p_j (1 - \tau_{ij}) + m_i}
\end{equation}
and $B$ be the $n_s \times n_s$ matrix with $0$ on the main diagonal and rows containing copies of
\begin{equation}
    B_{i} = - \frac{k d \sum_j p_j (1 - \tau_{ij})}{kd \sum_j p_j (1 - \tau_{ij}) + m_i}.
\end{equation}

Then, if $det(I - B) \neq 0$, the global open-access equilibrium exists and is
\begin{equation}
    S = (I - B)^{-1}A.
\end{equation}
\end{restatable}

\begin{restatable}[Reduction]{lemm}{oaReduction}
\label{lemm:reduction}
For every $n_s$-player global open-access equilibrium $S^\star$, there exists a $2$-player global open-access equilibrium $S^{\star,2}$ with players $i$ and $j$ such that 
\begin{align}
    S_i^\star &= S_i^{\star,2}, \\
    S_{-i}^\star &= S_j^{\star,2}.
\end{align}
\end{restatable}

The equilibrium satellite distribution with two spacefaring nations is shown in equation \ref{eqn:two_nation_eqm}:
\begin{equation}
    (S_i^\star, S_j^\star) = \left( \sigma_i r_i , ~ \sigma_j r_j \right),
    \label{eqn:two_nation_eqm}
\end{equation}

where
\begin{align}
    \sigma_i & \equiv \frac{(1+k(Q-D_0))(1-kdr_j)}{kd \sum_{j=1}^{n_s} p_j (1 - \tau_{ij}) + m_i}, \\
    r_i & \equiv \frac{\sum_{j=1}^{n_s} p_j(1-\tau_{ij})}{kd\sum_{j=1}^{n_s} p_j(1-\tau_{ij}) + m_i}.
\end{align}

$\sigma_i$ reflects the benefit-cost ratio of marginal debris abatement---the two terms in the numerator reflect the effects of abatement itself ($1+k(Q-D_0)$) and the effects of additional collision risk due to others ($1 - kdr_j$). $\sigma_i$ has percentage units. $r_i$ reflects the benefit-cost ratio of operating another satellite, accounting for the costs of forgone revenues due to collision risk ($kd\sum_{j=1}^{n_s} p_j(1-\tau_{ij})$). It has units of satellites, and can be interpreted as the maximum economically-sustainable satellite sector size for nation $i$. Lemma \ref{lemm:decomposition} establishes that the decomposition in equation \ref{eqn:two_nation_eqm} is valid even in the $n_s$-player case.\\

\begin{restatable}[Decomposition]{lemm}{oaDecomposition}
\label{lemm:decomposition}
The global open-access equilibrium satellite distribution can be decomposed into $S = \sigma \odot r$, where $\sigma$ and $r$ are $n_s \times 1$ vectors, $\odot$ is the Hadamard product, and only $\sigma$ depends on $Q$.
\end{restatable}

As long as $r_i, r_j < (kd)^{-1}$, then $\sigma_i, \sigma_j > 0$. The global open-access equilibrium in this case exists as long as $r_i r_j \neq (kd)^{-2}$. Since $m_i, m_j > 0 \implies r_i, r_j < 1$, both requirements can be satisfied when $m_i, m_j > 0$ and $k$ is small enough. These conditions are important for understanding the intuition of the results that follow, so we describe them briefly. $kd$ is the probability of satellite destruction from the lifetime debris created by a single new satellite, and has units of probability per satellite. The condition $r_i < (kd)^{-1}$ therefore states that the long-run maximum economically-sustainable satellite fleet maintainable by nation $i$ be less than the number of satellites destroyed per incremental satellite launch. If this condition were violated for any nation $i$ then the open-access equilibrium fleet size for nation $j$ would be zero. To see this, note that
\begin{align}
    r_i \geq (kd)^{-1} \implies& 0 \geq 1 - kd r_i \\
    \implies & \sigma_j = 0, S^*_j = 0.
\end{align}

Imposing $r_i < (kd)^{-1}$  for all $i$ therefore implies that the economic net benefits of orbit use for all national satellite sectors are not so large that they will naturally (i.e. due to purely-economic considerations) crowd all other nations out of orbital space. If there were a nation which crowded all others out in this fashion, the questions of this paper would reduce to questions of regulatory policy within that nation. Such questions are addressed in the existing economic literature, e.g. \cite{rouillon2020physico, rao2020orbital}, so we impose this condition here. We state this formally in Assumption \ref{assn:no_natural_orbital_seizure}.\\

\begin{assn}
\label{assn:no_natural_orbital_seizure}
The economic net benefits of orbit use for all national satellite sectors are not so large that any nation will crowd all other nations out of orbital space due to purely economic considerations, i.e.
\begin{equation}
r_i < (kd)^{-1} ~ \forall i \nonumber .
\end{equation}
\end{assn}

While not necessarily dispositive, evidence for Assumption \ref{assn:no_natural_orbital_seizure} being satisfied can be seen in that commercial satellite fleets are operated from multiple nations.\footnote{For example, Planet (a US company) and ICEYE (a Finnish company) both operate remote sensing satellite fleets, albeit of different sizes.} \\

We establish three key results about open-access orbit use under national regulations and international agreements in Propositions \ref{prop:taxes_rents}, \ref{prop:abatement_growth}, and \ref{prop:treaty_effective}. First, taxes targeted at a particular fleet are effective at reducing the targeted satellite fleet size, but do so at the cost of increasing the size of other satellite fleets---an ``open-access tax rebound effect''. This rebound effect follows from equation \ref{eqn:sat_bestresponse_static}: satellite fleets compete with each other for rents from using orbital space, and one fleet's withdrawal releases rents for others to claim. This rebound effect can be interpreted as international capital mobility from more-taxed to less-taxed jurisdictions in the spirit of \cite{oates1999essay} or simply as foreign sectoral expansion. Under mild conditions the open-access tax rebound is smaller than $1$ in magnitude, i.e. a satellite tax increase \emph{will} reduce the total number of satellites in orbit despite an open-access rebound.\\

\begin{restatable}[Taxes and rent dissipation]{prop}{taxesRents}
\label{prop:taxes_rents}
Under Assumption \ref{assn:no_natural_orbital_seizure},
\begin{align}
    \forall i,j, ~~ \pdv{S_i^\star}{\tau_{ij}} &< 0, \\
    \forall i \neq j, ~~ \pdv{S_j^\star}{\tau_{ij}} &> 0,
\end{align}
and $\forall i \neq j$,
\begin{align}
    %(CONDITION) \implies  
    - \pdv{S_i^\star}{\tau_{ij}}  >  \pdv{S_j^\star}{\tau_{ij}}.
\end{align}
\end{restatable}

Corollary \ref{coro:taxes_reduce_required_abatement} follows immediately.\\

\begin{coro}
\label{coro:taxes_reduce_required_abatement}
Given Proposition \ref{prop:taxes_rents}, satellite taxes reduce the long-run debris stock and therefore the abatement level required to avert catastrophe. Formally,
\begin{equation}
    \label{eqn:taxes_reduce_required_abatement}
    \pdv{\bar{Q}}{\tau_{ij}} = d\left( \pdv{S^\star_i}{\tau{_{ij}}} + \pdv{S^\star_j}{\tau{_{ij}}} \right) <0.
\end{equation}
\end{coro}

Second, greater debris abatement increases the size of every nation's satellite sector, with sectors facing lower tax rates growing more than sectors facing higher tax rates. Intuitively, debris abatement increases the amount of usable orbital volume. Open-access sectors take advantage of this by launching more satellites. Sectors facing lower tax rates are then able to profitably expand relatively more (though potentially still less in absolute terms) than sectors facing higher tax rates.\\

\begin{restatable}[Debris abatement and sectoral expansion]{prop}{abatementGrowth}
\label{prop:abatement_growth}
Under Assumption \ref{assn:no_natural_orbital_seizure}
\begin{align}
    \forall i, ~~ \pdv{S_i^\star}{Q} &> 0, \\ 
    \forall i,j, ~~ \frac{\partial^2 S_i^\star}{\partial Q \partial \tau_{ij}} &< 0
\end{align}
\end{restatable}

Third, in order for greater debris abatement from a treaty to actually reduce the stock of debris in orbit despite open-access responses inducing additional launches, the marginal collision risk from a new satellite must be bounded by a constant (stated precisely in Assumption \ref{assn:effective_abatement_bound}). Intuitively, if the marginal collision risk is sufficiently small then the rent-dissipating behavior of Propositions \ref{prop:taxes_rents} and \ref{prop:abatement_growth} won't fully erode the gains from greater debris abatement.\\

\begin{assn}
\label{assn:effective_abatement_bound}
The marginal collision risk from a new satellite is bounded by a constant, i.e.
\begin{equation}
kd < \frac{1}{2}.
\end{equation}
\end{assn}

\begin{restatable}[Effective abatement treaties]{prop}{treatyEffective}
\label{prop:treaty_effective}
Under Assumptions \ref{assn:no_natural_orbital_seizure} and \ref{assn:effective_abatement_bound} and if $m_i > 0 ~~ \forall i$, then
\begin{equation}
    \pdv{D^\star}{Q} < 0,
\end{equation}
where
\begin{equation}
    D^\star = d \sum_{i}S^\star + D_0 - Q.
\end{equation}
\end{restatable}

Assumption \ref{assn:effective_abatement_bound} is not very restrictive, and is likely satisfied now and in the future. \cite{rao2020orbital} offer supporting evidence on this point. Using a calibrated physico-economic model, they show that the rebound effect from debris removal that costs satellite operators nothing (analogous to abatement here) generally does not induce a launch response strong enough to overturn the benefits of debris removal.

\subsection{National regulations}

Each nation $j$ receives benefits from satellite services provided by all sectors per their market access, subject to open-access behavior by satellite operators. The benefits received by nation $j$ are shown in equation \ref{eqn:national_benefits}:
\begin{equation}
  W_j =   (1 - kD) p_j \sum_{i=1}^{n_s} (1 - \tau_{ij})  S_i^\star, % + \alpha_j \sum_{i=1}^{n_s} \tau_{ij}, 
  \label{eqn:national_benefits}
\end{equation}

where $S_i^\star$ satisfies equation \ref{eqn:open_access}. Each nation $j$ sets its virtual tax schedule, $\tau_{\cdot j}$ to maximize the benefits it receives from satellite services:
\begin{align}
    \max_{\tau_{\cdot j}} ~&~  W_j = (1 - kD) p_j \sum_{i=1}^{n_s} (1 - \tau_{ij})  S_i^\star \label{eqn:tax_maxprogram} \\
    \text{s.t.} &~ S_i^\star = \frac{\sum_j p_j (1 - \tau_{ij}) (1 - k(dS_{-i} + D_0 - Q))}{kd\sum_j p_j (1 - \tau_{ij}) + m_i} ~~ \forall i, \nonumber \\
    &~ D = d \sum_{i=1}^{n_s} S_i^\star + D_0 - Q. \nonumber
\end{align}

We do not include the catastrophe damages in program \ref{eqn:tax_maxprogram} to reflect that nations acting on their own do not internalize the full long-run benefits of orbital protection. We also hold the abatement level fixed to identify the incentives to tax in the absence of agreements regarding abatement. A typical first-order condition characterizing the optimal tax rate imposed by nation $j$ on sector $i$, $\tau_{ij}^\star(\tau_{\cdot \cdot})$, has the form in equation \ref{eqn:tax_foc}:
\begin{equation}
\hspace*{-1.25cm}
    \tau^\star_{ij}(\tau_{\cdot \cdot}) : \underbrace{-k d p_j \sum_{i=1}^{n_s} (1 - \tau_{ij})  S_i^\star \left[ \pdv{S_i^\star}{\tau_{ij}} + \pdv{S_{j}^*}{\tau_{ij}} \right]}_{\text{\shortstack{Value of services gained due to\\cleaner orbital environment\\(environmental cleanup channel)}}} + \underbrace{(1 - kD)p_j (1 - \tau_{jj}) \pdv{S_j^\star}{\tau_{ij}}}_{\text{\shortstack{Value of services gained due to\\sector $j$ responding to\\higher taxes on $i$\\(fleet expansion channel)}}} = \underbrace{(1 - kD)p_j \left[ S_i^\star - (1 - \tau_{ij}) \pdv{S_i^\star}{\tau_{ij}} \right]}_{\text{\shortstack{Value of services lost due to\\higher taxes on $i$\\(fleet reduction channel)}}}. 
    \label{eqn:tax_foc}
\end{equation}

The two terms on the left-hand side of equation \ref{eqn:tax_foc} describe sources of national gains through higher taxes on satellite sector $i$. There are two channels here: an environmental cleanup channel and a fleet expansion channel. First, when nation $j$ taxes sector $i$ (i.e. any foreign sector), sector $i$ reduces its fleet size while sector $j$ (i.e. the domestic sector) increases its fleet size. By Proposition \ref{prop:taxes_rents}, the increase in $S_j^\star$ is smaller than the decrease in $S_i^\star$. Increasing $\tau_{ij}$ therefore reduces orbital collision risk and increases the value of satellite services $j$ receives. Second, the expansion of sector $j$'s fleet due to a higher tax on sector $i$'s fleet provides some (gross) benefits to nations receiving satellite services. The term on the right-hand side of equation \ref{eqn:tax_foc} describes the national losses through higher taxes on satellite sector $i$, i.e. the value lost due to reductions in $S_i^\star$.\footnote{If some of these regulations take the form of explicit revenue-generating taxes, the left-hand size of equation \ref{eqn:tax_foc} will include an additional term for tax revenues gained from a higher rate on $i$ and a larger tax base for $j$, while the right-hand size will include an additional term for tax revenues lost from a smaller tax base in fleet $i$. We abstract from these issues to focus on the environmental considerations.}  A debris abatement treaty which assured sufficient abatement to avoid catastrophe would enter here to make the environmental cleanup channel larger (since $\pdv{\bar{Q}}{\tau_{ij}} < 0$ from Corollary \ref{coro:taxes_reduce_required_abatement}). A treaty which assured only a fixed level of debris abatement would have no effect on the environmental cleanup channel. Any level of debris abatement would reduce the magnitude of both the fleet expansion and fleet reduction channels, while larger legacy debris stocks would have the opposite effect.\\

In a \textit{global regulatory equilibrium}, all nations set their tax rates to maximize their benefits received from satellites under open access to orbit. Definition \ref{def:tax_eqm} formalizes this.\\

\begin{defn}(Global regulatory equilibrium)
\label{def:tax_eqm}
A global regulatory equilibrium is a schedule of satellite taxes $\{\tau_{ij}^\star\}$ which, given a distribution of open-access satellite sector sizes $\{S^\star_i\}$, net national debris abatement $Q$, a legacy debris stock $D_0$, market sizes $\{p_j\}$, sector efficiencies $\{m_i\}$, and technologies $(k,d)$, solve program \ref{eqn:tax_maxprogram} for all $j$.
\end{defn}

Assumption \ref{assn:tax_improves_welfare} shows a condition for satellite taxes on nation $i$'s fleet to improve nation $j$'s welfare.\\

\begin{assn}
\label{assn:tax_improves_welfare}
Define the semi-elasticity of the total satellite fleet with respect to a change in $\tau_{ij}$,
\begin{equation}
\mathcal{S}^\star_{ij} \equiv \frac{\pdv{S^\star_i}{\tau_{ij}} + \pdv{S^\star_j}{\tau_{ij}}}{S^\star_i + S^\star_j}.
\end{equation}
The legacy debris stock, $D_0$, is large enough that Condition \ref{eqn:tax_improves_welfare_condition} holds for nation $i$:
\begin{equation}
    \label{eqn:tax_improves_welfare_condition}
    -kd \mathcal{S}^\star_{ij} 
    %+ \pdv{Q}{\tau_{ij}} 
    > (1-kD) \frac{S^\star_i}{S^\star_i + S^\star_j} - (1-kD)\mathcal{S}^\star_{ij}.
\end{equation}
\end{assn}

Under Proposition \ref{prop:taxes_rents} the semi-elasticity in Assumption \ref{assn:tax_improves_welfare} is negative, making the left-hand side of condition \ref{eqn:tax_improves_welfare_condition} positive. The right-hand side of condition \ref{eqn:tax_improves_welfare_condition} is also positive.
However, if the legacy debris stock is large enough, the right-hand side of condition \ref{eqn:tax_improves_welfare_condition} will be small. At sufficiently large levels of $D_0$ Assumption \ref{assn:tax_improves_welfare} can be satisfied for all nations. Identifying nations for which Assumption \ref{assn:tax_improves_welfare} holds currently is an important question for future empirical research. \\

Regardless, Proposition \ref{prop:eqm_tax_distn} shows that Assumption \ref{assn:tax_improves_welfare} holding for all nations is a sufficient but not necessary condition for existence a global regulatory equilibrium in which all nations maintain positive tax rates on all satellite fleets from which they receive services. National incentives to have larger domestic satellite fleets tilt the equilibrium in favor of taxing others' fleets more than one's own, but in general do not overturn the result.\\

\begin{restatable}[Equilibrium tax distribution]{prop}{eqmTaxDistn}
\label{prop:eqm_tax_distn}
If Assumptions \ref{assn:no_natural_orbital_seizure}, \ref{assn:effective_abatement_bound}, and \ref{assn:tax_improves_welfare} hold for all $i, j$,
then there exists a global regulatory equilibrium where all nations levy positive taxes on all satellite sectors, i.e.
\begin{equation}
    \forall i, j, ~~ \tau^\star_{ij} > 0.
\end{equation}
\end{restatable}

Intuitively, the larger the legacy debris stock the greater the loss of satellite services due to debris-related collisions. This can lead nations to regulate satellites in order to increase the value they derive from the orbital environment. If the legacy debris stock is large enough, this incentive applies to all nations and Proposition \ref{prop:eqm_tax_distn} follows.

\subsection{International treaties}
\label{sec:model_treaties}

The marginal benefit of debris abatement for nation $i$ is the change in the amount of satellite services $i$ receives in response to an increase in debris abatement.\footnote{We only impose open access to orbit and not optimal satellite taxation at this stage. This serves two purposes. First, by imposing open access we can account for how the treaty will alter the relative sizes of national satellite sectors through capital mobility and economic expansion. In effect we are assuming that the treaty discussed here will not undermine the Outer Space Treaty. Second, by not imposing optimal taxation, we can study how the treaty will interact with an arbitrary distribution of binding national satellite regulatory policies. This is more flexible than imposing optimal taxation conditions, and mirrors our approach in the previous section.} That is,
\begin{equation}
    b_i(Q) \equiv \pdv{W_i}{Q}.
\end{equation}

$b_i(Q)$ represents each nation's marginal benefit from abating non-catastrophic debris accumulation. Since each nation may receive different benefits from satellite services, their marginal benefit of debris abatement may also vary.\footnote{Non-spacefaring nations can contribute to debris abatement, e.g. by purchasing active debris removal services.} However, since debris abatement confers non-rival and non-excludable benefits, nation $i$ benefits from the total amount of debris abatement $Q$ and not just their individual contribution $q_i$, where 
\begin{equation}
Q = \sum_{i=1}^N q_i.
\end{equation}

Classifying nations into two groups (nation $i$ and all others) and using the reduction from Lemma \ref{lemm:reduction} to simplify the satellite fleet responses to the $2$-nation case, $b_i(Q)$ has the form shown in equation \ref{eqn:two_nation_abatement_benefits}:
\begin{equation}
\label{eqn:two_nation_abatement_benefits}
    b_i(Q) = \left( \frac{ \sum_j p_j(1 - \tau_{ij})  }{kd \sum_j p_j (1 - \tau_{ij}) + m_i} \right)^2  \left( \frac{k (1 - kd r_j)}{1 - (kd)^2 r_i r_j} \right) \left( (1 - 2\sigma_i(Q)) kd \sum_j p_j (1 - \tau_{ij}) + m_i \right) .
\end{equation}

Notice the presence of $\sigma_i$, with dependence on $Q$ made explicit. Since $\sigma_i$ is linear in $Q$, the marginal benefits of debris abatement can be rewritten as a linear function, shown in equation \ref{eqn:two_nation_abatement_benefits_linear}:
\begin{align}
\label{eqn:two_nation_abatement_benefits_linear}
    b_i(Q) &= \alpha_i - \beta_i Q,
\end{align}

where
\begin{align}
    \alpha_i &= \left( \frac{ \sum_j p_j(1 - \tau_{ij})  }{kd \sum_j p_j (1 - \tau_{ij}) + m_i} \right)^2 \left( \frac{k (1 - kd r_j)^2 }{1 - (kd)^2 r_i r_j} \right) ,\\
    \beta_i &= 2k \left( \frac{ \sum_j p_j(1 - \tau_{ij})  }{kd \sum_j p_j (1 - \tau_{ij}) + m_i} \right)^2  \left( \frac{k (1 - kd r_j)^2}{1 - (kd)^2 r_i r_j} \right) \left( \frac{kd \sum_j p_j (1 - \tau_{ij})}{kd \sum_j p_j (1 - \tau_{ij}) + m_i} \right).
\end{align}

The literature on international treaties often treats the marginal benefits of pollution abatement as affine and increasing in total abatement, e.g. \cite{singer2011international,barrett2013climate, barrett2014sensitivity}. Here, a linear and decreasing form emerges from the value a nation receives from satellite services in the prior stages of the game. Under Assumption \ref{assn:no_natural_orbital_seizure}, both coefficients of equation \ref{eqn:two_nation_abatement_benefits_linear} are positive, making the marginal benefits of abatement decreasing in the total amount abated.\footnote{This property is similar to the form used in \cite{gelves2016international}, though the benefits function there is quadratic and concave.} Intuitively, a decreasing marginal benefits function can emerge from prioritizing the most valuable types of abatement (e.g. removing the objects posing the greatest fragmentation risk or passivating spent launch vehicles) and implementing them first. \\

We assume there exists a common set of debris abatement technologies available to all nations, such that abating $q_i$ units in the cost-minimizing pattern has marginal cost $\frac{c}{2} q_i^2$. The payoff to nation $i$ from abating $q_i$ units and experiencing $Q$ units of total abatement, ignoring catastrophic damages, are shown in equation \ref{eqn:abatement_payoff_component}:
\begin{equation}
\label{eqn:abatement_payoff_component}
    b_i (Q) - \frac{c}{2} q_i^2.
\end{equation}

We assume that the economic incentive to use orbital space under open access, despite any satellite taxes, is sufficient to cause Kessler Syndrome in the absence of debris abatement, i.e. $dS^\star > \bar{D}$. Since the payoffs of abatement are quadratic, when nations avert catastrophe they will do so by providing total abatement exactly equal to $\bar{Q}$. Recall that if catastrophe is not averted (i.e. debris exceeds $\bar{D}$, or abatement is less than $\bar{Q}$), then all nations experience damages of $X$. This gives the abatement payoff to nation $i$ shown in equation \ref{eqn:abatement_payoff}:
\begin{align}
\label{eqn:abatement_payoff}
    \pi_i = \begin{cases}
    b_i (\bar{Q}) - \frac{c}{2} q_i^2 &\text{ if } Q \geq \bar{Q} \\
    b_i (Q) - X - \frac{c}{2} q_i^2  &\text{ if } Q < \bar{Q}.
    \end{cases}
\end{align}

There are two interpretations of debris abatement: reduction of debris produced by satellites launched, and removal of debris in orbit. Reducing debris production through a treaty is relatively straightforward: each nation agrees to reduce debris production by an agreed-upon amount. Removing debris is more complicated. We assume that the treaty allows debris placed by any nation's operators to be removed by any other nation's operators. In practice, rights to debris objects are held by the entities which launched them, creating non-trivial transactions costs to debris removal \citep{weeden2011overview}. We maintain this assumption to illustrate the scope of possible outcomes if such a treaty is signed, or if institutional reforms (e.g. debris salvage rights as described in \cite{munoz2018regulating}) facilitate the formation of competitive markets for debris rights.  \\

When nations supply abatement to maximize their individual payoffs there are two symmetric pure-strategy Nash equilibria: one where nations ignore the risk of catastrophe and supply no abatement ($q_i = 0$) and one where they equally share the burden of avoiding the catastrophe threshold ($q_i = \frac{\bar{Q}}{N}$).
This is formalized in Proposition \ref{prop:eqm_debris_abatement}. \\

\begin{restatable}[Nash abatement levels]{prop}{nashAbatement}
\label{prop:eqm_debris_abatement}
When 
\begin{equation}
\label{eqn:no_treaty_defection}
    X \geq \beta_i \frac{\bar{Q}}{N} + \frac{c(\bar{Q}/N)^2}{2},
\end{equation}
 for all $i$, then there are at least two pure-strategy Nash equilibrium debris abatement levels,
\begin{align}
    q_i &= 0 ~~~ \forall i, \\
    q_i &= \frac{\bar{Q}}{N} ~~~ \forall i.
\end{align}
Only $q_i = \frac{\bar{Q}}{N}$ for all $i$ averts catastrophe.
\end{restatable}

Condition \ref{eqn:no_treaty_defection} is a lower bound on the catastrophe damages such that no nation would prefer to be pivotal in causing catastrophe. It states that the catastrophe damages are greater than the average marginal payoff from being the one nation supplying insufficient abatement.\footnote{National security concerns could matter here. For example, if nation $i$ preferred to see satellites operated from nation $j$ destroyed by debris despite the risk to their own satellites, then the individual benefits term $b_i$ could be increasing in the debris stock. Particularly if $b_i$ is relatively flat in the amount of satellite services received (e.g. because $p_i$ is small), then $i$ may prefer to incur catastrophe precisely because $j$ prefers not to. Exploring such interdependent preferences are beyond our scope here but is an important area for future research.} \\

The Nash equilibria identified in Proposition \ref{prop:eqm_debris_abatement} reflect the fact that debris abatement is a public good. Each piece removed confers non-rival and non-excludable benefits to all orbit users, though the cost is borne privately. The catastrophe-averting equilibrium can be sustained if and only if each nation faces a sufficiently-low incentive to free-ride. Condition \ref{eqn:no_treaty_defection} offers some intuition for how the incentive to cooperate emerges without a treaty. As in \cite{barrett2013climate}, the damages from the catastrophic outcome act as Nature's incentive to cooperate---if the damages are sufficiently severe, nations will seek to avert Kessler Syndrome even without a formal treaty to coordinate them. But as discussed in \cite{barrett2016coordination}, coordination which involves reshaping incentives (rather than relying on traditional enforcement mechanisms) nearly always works better than voluntarism. Definition \ref{defn:self_enforcing_treaty} describes the type of ``self-enforcing'' treaty we consider, and Proposition \ref{prop:self_enforcing_treaty} describes the treaty structure.

\begin{defn}(Self-enforcing treaty)
\label{defn:self_enforcing_treaty}
A self-enforcing treaty to prevent Kessler Syndrome is a treaty which
\begin{enumerate}
    \item ensures total abatement supply is $Q=\bar{Q}$, and
    \item prevents nations who prefer to avert Kessler Syndrome from leaving the treaty, by using the abatement responses of the remaining signatories.
\end{enumerate}
\end{defn}

\begin{restatable}[Self-enforcing treaty]{prop}{selfEnforcingTreaty}
\label{prop:self_enforcing_treaty}
A debris abatement treaty with $N$ signatories can avert Kessler Syndrome despite defections if, in response to defections by any nation $i$, the treaty instructs the remaining signatories to supply
\begin{equation}
\label{eqn:treaty_response_abatement}
    Q_{\setminus i} = \bar{Q} + \frac{1}{c} \left( \beta_i - \sqrt{\beta_i^2 + 2cX} \right)
\end{equation}
units of debris abatement. If 
\begin{equation}
\label{eqn:defector_gets_punished}
    \frac{\bar{Q}}{N} < - \frac{1}{c} \left( \beta_i - \sqrt{\beta_i^2 + 2cX} \right)
\end{equation}
also holds for all $i$, then the treaty will also be self-enforcing.
\end{restatable} 

As in Proposition 2 of \cite{barrett2013climate}, the analysis in Proposition \ref{prop:self_enforcing_treaty} assumes full participation. Yet \cite{singer2011international} finds that a debris abatement treaty is likely to support only a lower level of participation, albeit without accounting for open access to orbit. Can a debris abatement treaty sustain full participation once open access to orbit is accounted for? There are two ways to resolve this question affirmatively. First, if condition \ref{eqn:no_treaty_defection} holds, then nation $i$ would prefer averting catastrophe rather than incurring it, and remaining within the treaty offers a way to do so. Second, provided condition \ref{eqn:defector_gets_punished} holds and the treaty response is a punishment for all potential defectors rather than a burden-reduction, no nation will prefer to defect. Full participation then becomes a self-reinforcing (as well as self-enforcing) equilibrium. The defector will recognize that if they leave the treaty the remaining signatories will exploit the defector's aversion to Kessler Syndrome and force them to incur greater abatement costs than if they remained in the treaty. Condition \ref{eqn:no_treaty_defection} shows that all nations will prefer to avert Kessler Syndrome if either the damages are high enough or the total abatement costs are low enough.  \\

As abatement technologies improve and $c \to 0$, the remaining signatories' response to defection in equation \ref{eqn:treaty_response_abatement} becomes a punishment rather than a burden-reduction for all potential defectors. Similarly, inspection of condition \ref{eqn:no_treaty_defection} reveals that innovations which reduce the cost of debris abatement make averting Kessler Syndrome (even without a treaty) more attractive. Technology improvements which reduce the cost of abatement make both remaining within the treaty and treaty-free coordination more attractive. It is less clear how changes in national taxes ($\tau_{\cdot \cdot}$) or $i$'s satellite sector efficiency ($m_i$) will affect the incentive to cooperate through the parameter $\beta_i$. Lemma \ref{lemm:beta_i_statics} establishes these comparative statics, showing that $\beta_i$ is decreasing in $i$'s satellite sector efficiency and all taxes. \\

\begin{restatable}[Comparative statics]{lemm}{betaiCompStatics}
\label{lemm:beta_i_statics}
$\beta_i$ is decreasing in $m_i$ and $\tau_{\cdot \cdot}$.
\end{restatable}

Using Lemma \ref{lemm:beta_i_statics}, we can now establish our key result about treaties and taxes: that national satellite tax policies make an international debris abatement treaty easier to sustain. This result operates through two channels. First, higher national satellite tax rates make it more likely that all nations prefer to avert catastrophe even without a treaty. Second, provided the costs of debris abatement are low enough (or the damages from Kessler Syndrome are high enough), higher national satellite tax rates make it more likely that defection from the treaty is a punishment rather than burden-reduction. \\

\begin{restatable}[Taxes support treaties]{prop}{taxesSupportTreaties}
\label{prop:taxes_support_treaties}
Higher foreign taxes on nation $i$'s satellite sector give nation $i$ a stronger incentive to avert Kessler Syndrome. If the costs of debris abatement are low enough (or the damages from Kessler Syndrome are high enough), then higher foreign taxes on $i$ also give $i$ a stronger incentive to remain within a self-enforcing debris abatement treaty. Formally,
\begin{align}
\label{eqn:tax_makes_nash_aversion_better}
    \pdv{}{\tau_{ij}} \left( \beta_i \frac{\bar{Q}}{N} + \frac{c(\bar{Q}/N)^2}{2} \right) &< 0,
\end{align}
and if
\begin{equation}
    \sqrt{\beta_i^2 + 2cX} - c \beta_i > 0,
\end{equation}
then
\begin{align}
\label{eqn:tax_makes_defection_worse}
    \pdv{}{\tau_{ij}} \left(- \frac{\bar{Q}}{N} - \frac{1}{c} \left( \beta_i - \sqrt{\beta_i^2 + 2cX} \right) \right) &> 0.
\end{align}
\end{restatable}

A similar conclusion follows for improvements in sector efficiencies.
\section{Discussion}
\label{sec:discussion}

Many scholars analyzing orbital-use management have established pessimistic or concerning results. These include the high likelihood that physical tipping points have been or are soon to be crossed \citep{krisko2001critical, lioujohnson2008_soi, kessler2010kessler, drmola2018kessler}, the economic incentives provided by current management institutions to overuse orbital space and induce Kessler Syndrome \citep{adilov2015economic, rouillon2020physico, rao2020orbital, rao2022cost}, and the lack of incentives (both due to legal and policy barriers as well as the natural economic structure of the problem) to abate or remove orbital debris without free-riding \citep{weeden2011overview, adilov2015economic, rao2018economic, klima2018space}. \\

This paper offers more optimistic results. Proposition \ref{prop:eqm_tax_distn} shows that a equilibrium can exist where nations regulate their own and each others' satellites without transfers of surplus. Proposition \ref{prop:self_enforcing_treaty} provides guidance on how to draft a debris abatement treaty which can avert Kessler Syndrome without requiring supranational enforcement beyond incentive-compatible responses specified in the treaty itself. Proposition \ref{prop:taxes_support_treaties} shows how bilateral national satellite taxes can increase the likelihood of treaty success. To the extent that debris abatement treaties can be linked to other environmental issues (e.g. climate treaties) during negotiations, there is even more scope for treaty-based debris abatement \citep{adilov2022understanding}. \\

We decompose the national incentives to use ``bilateral satellite taxes'' (i.e. satellite regulatory policies imposed as conditions for market access) into three channels. The first two channels, environmental cleanup and fleet expansion, create an incentive for nations to levy satellite taxes on other nations' satellites. The third channel, fleet reduction, limits that incentive from inducing an international ``tax war''. These channels suggest that even in the absence of catastrophe risk, nations will prefer to levy some rather than no satellite taxes on each others' fleets so as to control collision risk and increase the long-run value of satellite services received.  This echoes an important result in \cite{brander1998open}, where trade in a renewable resource subject to open access can reduce welfare by hastening stock depletion and tariffs or trade restrictions can improve welfare by mitigating the negative externality. In our model satellite services are a tradeable good produced using orbital space, a renewable resource. \\

The magnitude of global welfare gains from a globally-harmonized optimal orbital-use fee identified in \cite{rao2020orbital} suggests the existence of a global regulatory equilibrium with positive tax rates \emph{and} transfers between nations. Proposition \ref{prop:eqm_tax_distn} goes one step further, showing that a global regulatory equilibrium with positive satellite taxes can exist \emph{without} transfers and identifying a relevant sufficient condition. The sufficient condition requires debris abatement to be decreasing in tax rates, which can be assured by a debris abatement treaty subject to open-access responses. Together, Propositions \ref{prop:eqm_tax_distn} and \ref{prop:taxes_support_treaties} show that a system of non-harmonized national satellite taxes and an international treaty to abate just enough debris to avert Kessler Syndrome can reinforce each other, increasing the global benefits from orbit use without undermining the Outer Space Treaty.\\

We focus on the positive question of what equilibrium tax rates will look like when nations act strategically to maximize the economic benefits they receive from orbit use rather than normative questions of optimal fee system design. The literature on environmental federalism offers some insights regarding the efficiency properties of the system we identify. Orbital space is a fully spatially-connected resource, with different national markets for satellite services inducing preference heterogeneity. Open access and orbital debris create an interjurisdictional resource externality. Decentralized national policies (and the equilibrium identified in Proposition \ref{prop:eqm_tax_distn}) are therefore likely inefficient, albeit efficiency-enhancing, unless the marginal dynamic costs of debris are sufficiently high to make it privately optimal to totally eradicate debris and collision risk \citep{levinson2003environmental, costello2015partial, costello2017private}. This condition seems both unlikely to be satisfied---the external cost decomposition in \cite{rao2022cost} shows debris reductions will reduce the marginal dynamic cost of debris (see the third term of equation (30) there)---and technologically infeasible. The potential for catastrophe and an international treaty to avert it create an interjurisdictional fiscal externality \citep{wilson1999theories}. Proposition \ref{prop:taxes_support_treaties} can therefore be interpreted as a ``second-best'' result \citep{lipsey1956general}, wherein the fiscal and resource externalities mitigate each other. \\

While our analysis focuses on long-run possibilities and outcomes, our analysis offers suggestions on how to begin the process of implementing the interlocking system of orbital-use management we describe here. The key is the legacy debris stock. \cite{lemoine2020incentivizing} finds that large legacy carbon stocks make negative emissions policies more challenging to implement, a finding reaffirmed in our analysis of a self-enforcing treaty---notice the ways the required abatement level $\bar{Q}$ (which is increasing in the legacy stock $D_0$) appears in Propositions \ref{prop:eqm_debris_abatement} and \ref{prop:self_enforcing_treaty}. However, Proposition \ref{prop:eqm_tax_distn} shows that larger legacy debris stocks make a global regulatory equilibrium \emph{more} likely to emerge (i.e. larger $D_0$ makes Assumption \ref{assn:tax_improves_welfare} more likely to hold for all $i$). Further, Proposition \ref{prop:taxes_support_treaties} shows that bilateral taxes in the global regulatory equilibrium will support a self-enforcing debris abatement treaty. Thus, our analysis suggests that when the legacy debris stock is large, beginning with national regulatory policies may be a path toward implementing an international debris abatement treaty. \\

Technological advances which improve access to orbit and make it cheaper to abate debris also make it easier to sustain cooperation in preventing Kessler Syndrome. Policies which support these technologies seem like a useful path forward. Understanding how different types of satellite operators respond to different regulatory policies is also important in determining the shape of a potential global regulatory equilibrium. Empirical and experimental work in this direction seems like another useful direction to pursue.  \\

There remain many limitations to our analysis. Importantly, we do not consider national security and dual-use concerns. These are key drivers of orbit use. To the extent that nations prefer to weaken each others' satellite sectors---or satellite sectors serving each others' markets---our results overstate the potential for international cooperation and the benefits of international competition in orbit use. But to the extent that nations with high dependence on satellites for national security internalize the benefits of keeping the orbital environment clean, our results understate the potential for cooperation. Similar caveats apply in the political domain when applying our results to atmospheric carbon management. Further, anti-satellite missile testing may produce enough difficult-to-remove fragments that it is no longer possible to abate enough debris to prevent Kessler Syndrome. These limitations suggest several areas for future research. One avenue is to identify potential non-kinetic attacks on national satellite fleets which exploit open-access responses or national satellite tax policies so that these attack surfaces can be mitigated when designing the policy regime. Another is to explore the scope for decentralized or treaty-based deterrence or prevention of missile testing, while a third is to study mechanisms to overcome strategic non-economic orbital-use incentives. \\

We also ignore how the costs of debris abatement are financed. It is possible that such abatement can be financed by satellite taxes from the global regulatory equilibrium. This is an interesting avenue for future research. Yet we argue that focusing on budget-balanced tax-and-spend approaches to orbital-use management risks misunderstanding the nature of the problem. Open access to orbit is an extremely costly status quo, reducing the global value generated by satellites over the coming decades by an order of magnitude \citep{rao2020orbital}. This is not surprising---similar results have been observed in other open access resources, such as oil fields and fisheries \citep{libecap1984contractual, world2009sunken}. As debris accumulates, the damages incurred from open access will grow and eventually become costlier to address \citep{mcknight2010pay}. Allowing open-access orbit use to continue without economic policy responses is thus akin to allowing a depressed economy to be scarred by hysteresis. Fiscal policy to avert hysteresis can be self-financing through the additional growth it produces relative to the status quo \citep{delong2012fiscal}. Like fiscal policy in a depressed economy, satellite taxes and debris abatement pay for themselves. \\

Finally, our analysis also ignores uncertainty over the location of debris threshold for Kessler Syndrome. \cite{barrett2013climate} finds that such uncertainty can greatly limit the potential for self-enforcing international treaties to avert catastrophe. However, engineering analyses are continually refining estimates of the nature of the self-sustaining debris growth threshold, e.g. \cite{krisko2001critical, lioujohnson2008_soi, drmola2018kessler, lewis2020understanding}. In general, scientific understanding of the fundamental mechanisms and interactions driving orbital debris growth seems better than scientific understanding of the mechanisms and interactions driving planetary warming. Further, regardless of the location of the threshold, national regulatory policies can be used to support cooperation in debris abatement, and the incentives we identify to implement national regulatory policies are independent of the existence of a catastrophe threshold. Consequently, threshold uncertainty may not be as serious a threat to international cooperation in orbital-use management as in atmospheric carbon management.

\section{Conclusion}
\label{sec:conclusion}
% Summary. You’ve surely heard that when writing a research paper, “tell them what you’re going to tell them, tell them what you want to tell them, and tell them what you just told them.” This part is obviously tedious–you have just spent 40-some pages telling them–but it needs to be there, and it needs to be different enough from the abstract and the introduction. Note that I didn’t say it needs to say something new; it just needs to be different enough. If possible, tell a story.
Managing orbital space is critical to ensuring humanity continues to benefit from satellite services. While international pollution abatement agreements may be a preferred approach to managing a global commons, orbital-use management agreements may be more difficult to construct in the near future than national regulatory policies. In this paper we contribute to the literature and policy discussion around orbital-use management by examining the relationships between open access to orbit, national regulatory policies, and international treaties to prevent self-sustaining debris growth (``Kessler Syndrome''). We highlight three messages regarding international competition and cooperation in orbital-use management. \\

% Limitations. Some people like to have a “Limitations” section at the end of their results section; I like to have that myself. But the conclusion should also emphasize the limitations of your approach.
First, like a mobile capital stock subject to costly regulation in one jurisdiction, open access to orbit means that regulatory policies in one nation will shrink their satellite sector and induce expansion of other nations' satellite sectors. Despite this response pattern, bilateral regulatory policies will generally produce sustained improvements to the orbital environment and the international economic value derived from satellites. Second, a large legacy debris stock can induce all nations receiving satellite services to impose environmental regulations on all satellite fleets providing services as a condition for market access. Such a global regulatory equilibrium is compatible with incentives to maximize national benefits from satellite services. Third, a system of bilateral regulatory policies as conditions for market access can be used to support the formation and maintenance of a self-enforcing international treaty to abate debris production and avert Kessler Syndrome. Such a system can be incentive-compatible for all nations involved, avoiding the need for a supranational enforcer. \\
    
% Implications for Policy. Presumably, your work has some sort of implication for how policy is made in the real world. This will not always be the case–some papers make a purely technical point, or a point that is only ancillary when it comes to making other policy-related points–but I would guess that since you are reading this blog, there is a high likelihood that what you are working on has some policy implications. Discuss what those implications are, but don’t make claims that are not supported by your results, and try to assess the cost of what you propose in comparison to its benefits. You can do so somewhat imperfectly (if I were a betting man, I would bet that this is where the phrase “back-of-the-envelope calculation” comes up the most often in economics papers), since the point of your work was presumably about only one side of that equation–usually the benefits of something, sometimes its costs, but rarely both. In two or three sentences, can you identify the clear winners and losers of a given policy implications? Its political feasibility? How easy or hard it would be to implement?

There are many remaining open questions of orbital-use policy design, particularly relating to national security and dual-use technology concerns. Still, it is worth knowing that there is a light at the end of the tunnel. If economic benefits are a large enough share of the value nations derive from orbit use, international competition for orbital space and international cooperation to manage the orbital environment can be mutually reinforcing.

% Implications for Future Research. Finally, your work is not perfect. Your theoretical contribution could be generalized, or broadened by relaxing certain assumptions. Your empirical contribution could probably benefit from better causal identification for better internal validity. Even with a randomized controlled trial (RCT) with perfect compliance, you might want to run the same RCT in additional locations for external validity. If you are writing a follow-up paper, this is a good place to set the stage for it.*

\newpage

\bibliographystyle{aea}
\bibliography{bibliography}

@ARTICLE{lioujohnson2008_soi,
author= "J.C. Liou and N.L. Johnson",
title= "Instability of the present LEO satellite populations",
journal="Advances in Space Research",
volume="41",
issue="7",
pages="1046-1053",
month="January",
year=2008
}

@article{krisko2001critical,
  title={The critical density theory in LEO as analyzed by EVOLVE 4.0},
  author={Krisko, Paula H and Opiela, John N and Kessler, Donald J},
  journal={Space Debris},
  volume={473},
  pages={273--278},
  year={2001}
}

@article{weitzman2009modeling,
  title={On modeling and interpreting the economics of catastrophic climate change},
  author={Weitzman, Martin L},
  journal={The review of economics and statistics},
  volume={91},
  number={1},
  pages={1--19},
  year={2009},
  publisher={The MIT Press}
}

@article{highfill2022estimating,
  title={Estimating the United States Space Economy Using Input-Output Frameworks},
  author={Highfill, Tina C and MacDonald, Alexander C},
  journal={Space Policy},
  pages={101474},
  year={2022},
  publisher={Elsevier}
}

@article{van2021world,
  title={A World without Satellite Data as a Result of a Global Cyber-Attack},
  author={Van Camp, Charlotte and Peeters, Walter},
  journal={Space Policy},
  pages={101458},
  year={2021},
  publisher={Elsevier}
}

@article{levinson2003environmental,
  title={Environmental regulatory competition: A status report and some new evidence},
  author={Levinson, Arik},
  journal={National Tax Journal},
  volume={56},
  number={1},
  pages={91--106},
  year={2003},
  publisher={The University of Chicago Press}
}

@article{wilson1999theories,
  title={Theories of tax competition},
  author={Wilson, John Douglas},
  journal={National tax journal},
  volume={52},
  number={2},
  pages={269--304},
  year={1999},
  publisher={The University of Chicago Press}
}

@article{lipsey1956general,
  title={The general theory of second best},
  author={Lipsey, Richard G and Lancaster, Kelvin},
  journal={The review of economic studies},
  volume={24},
  number={1},
  pages={11--32},
  year={1956},
  publisher={JSTOR}
}

@article{bradley2009space,
  title={Space debris: Assessing risk and responsibility},
  author={Bradley, Andrew M and Wein, Lawrence M},
  journal={Advances in Space Research},
  volume={43},
  number={9},
  pages={1372--1390},
  year={2009},
  publisher={Elsevier}
}

@article{schaub2015cost,
  title={Cost and risk assessment for spacecraft operation decisions caused by the space debris environment},
  author={Schaub, Hanspeter and Jasper, Lee EZ and Anderson, Paul V and McKnight, Darren S},
  journal={Acta Astronautica},
  volume={113},
  pages={66--79},
  year={2015},
  publisher={Elsevier}
}

@article{beal2020taxing,
  title={Taxing congestion of the space commons},
  author={B{\'e}al, Sylvain and Deschamps, Marc and Moulin, Herv{\'e}},
  journal={Acta Astronautica},
  volume={177},
  pages={313--319},
  year={2020},
  publisher={Elsevier}
}

@article{o2019economic,
  title={Economic benefits of the global positioning system (GPS)},
  author={O'Connor, Alan C and Gallaher, Michael P and Clark-Sutton, Kyle and Lapidus, Daniel and Oliver, Zack T and Scott, Troy J and Wood, Dallas W and Gonzalez, Manuel A and Brown, Elizabeth G and Fletcher, Joshua},
  year={2019},
  publisher={RTI International}
}

@article{kessler2010kessler,
  title={The kessler syndrome: implications to future space operations},
  author={Kessler, Donald J and Johnson, Nicholas L and Liou, JC and Matney, Mark},
  journal={Advances in the Astronautical Sciences},
  volume={137},
  number={8},
  pages={2010},
  year={2010},
  publisher={Univelt, Inc.}
}

@article{adilov2022understanding,
  title={Understanding the Economics of Orbital Pollution Through the Lens of Terrestrial Climate Change},
  author={Adilov, Nodir and Alexander, Peter and Cunningham, Brendan},
  journal={Space Policy},
  pages={1014-71},
  year={2022},
  publisher={Elsevier}
}

@article{somma2019sensitivity,
  title={Sensitivity analysis of launch activities in Low Earth Orbit},
  author={Somma, Gian Luigi and Lewis, Hugh G and Colombo, Camilla},
  journal={Acta Astronautica},
  volume={158},
  pages={129--139},
  year={2019},
  publisher={Elsevier}
}

@article{arnas2021definition,
  title={Definition of Low Earth Orbit slotting architectures using 2D lattice flower constellations},
  author={Arnas, David and Lifson, Miles and Linares, Richard and Avenda{\~n}o, Mart{\'\i}n E},
  journal={Advances in Space Research},
  volume={67},
  number={11},
  pages={3696--3711},
  year={2021},
  publisher={Elsevier}
}

@article{unoosa2018european,
  title={‘European Global Navigation: Satellite System and Copernicus: Supporting the Sustainable Development Goals},
  author={UNOOSA},
  year={2018},
  publisher={United Nations Office for Outer Space Affairs Vienna}
}

@article{rouillon2020physico,
  title={A Physico-Economic Model of Low Earth Orbit Management},
  author={Rouillon, S{\'e}bastien},
  journal={Environmental and Resource Economics},
  volume={77},
  number={4},
  pages={695--723},
  year={2020},
  publisher={Springer}
}

@article{rao2022cost,
  title={Cost in Space: Debris and Collision Risk in the Orbital Commons},
  author={Rao, Akhil and Rondina, Giacomo},
  journal={Working paper. Latest draft available at \url{https://arxiv.org/abs/2202.07442}},
  year={2022}
}

@article{rao2018economic,
  title={Economic Principles of Space Traffic Control},
  author={Rao, Akhil},
  journal={Working paper. Latest draft available at \url{https://akhilrao.github.io/assets/working_papers/Economic_Principles_of_Space_Traffic_Control.pdf}},
  year={2019}
}

@article{adilov2015economic,
  title={An economic analysis of earth orbit pollution},
  author={Adilov, Nodir and Alexander, Peter J and Cunningham, Brendan M},
  journal={Environmental and Resource Economics},
  volume={60},
  number={1},
  pages={81--98},
  year={2015},
  publisher={Springer}
}

@article{rao2020orbital,
  title={Orbital-use fees could more than quadruple the value of the space industry},
  author={Rao, Akhil and Burgess, Matthew G and Kaffine, Daniel},
  journal={Proceedings of the National Academy of Sciences},
  volume={117},
  number={23},
  pages={12756--12762},
  year={2020},
  publisher={National Acad Sciences}
}

@article{brander1998open,
  title={Open access renewable resources: Trade and trade policy in a two-country model},
  author={Brander, James A and Taylor, M Scott},
  journal={Journal of International Economics},
  volume={44},
  number={2},
  pages={181--209},
  year={1998},
  publisher={Elsevier}
}

@techreport{guyot2021designing,
  title={Designing satellites to cope with orbital debris},
  author={Guyot, Julien and Rouillon, S{\'e}bastien},
  year={2021},
  institution={Groupe de Recherche en Economie Th{\'e}orique et Appliqu{\'e}e (GREThA). Working paper. Latest draft available at \url{http://bordeauxeconomicswp.u-bordeaux.fr/2021/2021-16.pdf}}
}

@article{weitzman1998far,
  title={Why the far-distant future should be discounted at its lowest possible rate},
  author={Weitzman, Martin L},
  journal={Journal of environmental economics and management},
  volume={36},
  number={3},
  pages={201--208},
  year={1998},
  publisher={Elsevier}
}

@article{delong2012fiscal,
  title={Fiscal policy in a depressed economy [with comments and discussion]},
  author={DeLong, J Bradford and Summers, Lawrence H and Feldstein, Martin and Ramey, Valerie A},
  journal={Brookings Papers on Economic Activity},
  pages={233--297},
  year={2012},
  publisher={JSTOR}
}

@techreport{lemoine2020incentivizing,
  title={Incentivizing Negative Emissions Through Carbon Shares},
  author={Lemoine, Derek},
  year={2020},
  institution={National Bureau of Economic Research}
}

@article{drmola2018kessler,
  title={Kessler syndrome: system dynamics model},
  author={Drmola, Jakub and Hubik, Tomas},
  journal={Space Policy},
  volume={44},
  pages={29--39},
  year={2018},
  publisher={Elsevier}
}

@article{wouters2016space,
  title={Space Debris Remediation, Its Regulation and the Role of Europe},
  author={Wouters, Jan and De Man, Philip and Hansen, Rik},
  journal={Eur. JL Reform},
  volume={18},
  pages={66},
  year={2016},
  publisher={HeinOnline}
}

@article{gilbert2021major,
  title={Major Federal Actions Significantly Affecting the Quality of the Space Environment: Applying NEPA to Federal and Federally Authorized Outer Space Activities},
  author={Gilbert, Alexander Q. and Vidaurri, Monica},
  journal={Environs: Environmental Law and Policy Journal},
  volume={44},
  isse={2},
  pages={233-272},
  year={2021},
  publisher={UC Davis}
}

@article{munoz2018regulating,
  title={Regulating the space commons: Treating space debris as abandoned property in violation of the outer space treaty},
  author={Mu{\~n}oz-Patchen, Chelsea},
  journal={Chi. J. Int'l L.},
  volume={19},
  pages={233},
  year={2018},
  publisher={HeinOnline}
}

@article{migaud2020protecting,
  title={Protecting Earth's Orbital Environment: Policy Tools for Combating Space Debris},
  author={Migaud, Michael R},
  journal={Space Policy},
  volume={52},
  pages={101361},
  year={2020},
  publisher={Elsevier}
}

@article{johnson2012application,
  title={Application of Ostrom’s Principles for Sustainable Governance of Common-Pool Resources to Near-Earth Orbit},
  author={Johnson-Freese, Joan and Weeden, Brian},
  journal={Global Policy},
  volume={3},
  number={1},
  pages={72--81},
  year={2012},
  publisher={Wiley Online Library}
}

@article{keohane1986reciprocity,
  title={Reciprocity in international relations},
  author={Keohane, Robert O},
  journal={International organization},
  volume={40},
  number={1},
  pages={1--27},
  year={1986},
  publisher={Cambridge University Press}
}

@article{weeden2011overview,
  title={Overview of the legal and policy challenges of orbital debris removal},
  author={Weeden, Brian},
  journal={Space Policy},
  volume={27},
  number={1},
  pages={38--43},
  year={2011},
  publisher={Elsevier}
}

@article{weeden2012taking,
  title={Taking a common-pool resources approach to space sustainability: A framework and potential policies},
  author={Weeden, Brian C and Chow, Tiffany},
  journal={Space Policy},
  volume={28},
  number={3},
  pages={166--172},
  year={2012},
  publisher={Elsevier}
}

@article{klima2018space,
  title={Space debris removal: Learning to cooperate and the price of anarchy},
  author={Klima, Richard and Bloembergen, Daan and Savani, Rahul and Tuyls, Karl and Wittig, Alexander and Sapera, Andrei and Izzo, Dario},
  journal={Frontiers in Robotics and AI},
  volume={5},
  pages={54},
  year={2018},
  publisher={Frontiers}
}

@article{barrett2013climate,
  title={Climate treaties and approaching catastrophes},
  author={Barrett, Scott},
  journal={Journal of Environmental Economics and Management},
  volume={66},
  number={2},
  pages={235--250},
  year={2013},
  publisher={Elsevier}
}

@article{barrett2014sensitivity,
  title={Sensitivity of collective action to uncertainty about climate tipping points},
  author={Barrett, Scott and Dannenberg, Astrid},
  journal={Nature Climate Change},
  volume={4},
  number={1},
  pages={36--39},
  year={2014},
  publisher={Nature Publishing Group}
}

@article{barrett2016coordination,
  title={Coordination vs. voluntarism and enforcement in sustaining international environmental cooperation},
  author={Barrett, Scott},
  journal={Proceedings of the National Academy of Sciences},
  volume={113},
  number={51},
  pages={14515--14522},
  year={2016},
  publisher={National Acad Sciences}
}

@article{gelves2016international,
  title={International environmental agreements with consistent conjectures},
  author={Gelves, Alejandro and McGinty, Matthew},
  journal={Journal of Environmental Economics and Management},
  volume={78},
  pages={67--84},
  year={2016},
  publisher={Elsevier}
}

@article{lewis2020understanding,
  title={Understanding long-term orbital debris population dynamics},
  author={Lewis, Hugh G},
  journal={Journal of Space Safety Engineering},
  volume={7},
  number={3},
  pages={164--170},
  year={2020},
  publisher={Elsevier}
}

@article{costello2015partial,
  title={Partial enclosure of the commons},
  author={Costello, Christopher and Qu{\'e}rou, Nicolas and Tomini, Agnes},
  journal={Journal of Public Economics},
  volume={121},
  pages={69--78},
  year={2015},
  publisher={Elsevier}
}

@article{costello2017private,
  title={Private eradication of mobile public bads},
  author={Costello, Christopher and Qu{\'e}rou, Nicolas and Tomini, Agnes},
  journal={European Economic Review},
  volume={94},
  pages={23--44},
  year={2017},
  publisher={Elsevier}
}

@article{singer2011international,
  title={An International Environmental Agreement for space debris mitigation among asymmetric nations},
  author={Singer, Michael J and Musacchio, John T},
  journal={Acta Astronautica},
  volume={68},
  number={1-2},
  pages={326--337},
  year={2011},
  publisher={Elsevier}
}

@article{millimet2013environmental,
  title={Environmental federalism: a survey of the empirical literature},
  author={Millimet, Daniel L},
  journal={Case W. Res. L. Rev.},
  volume={64},
  pages={1669},
  year={2013},
  publisher={HeinOnline}
}

@article{oates1999essay,
  title={An essay on fiscal federalism},
  author={Oates, Wallace E},
  journal={Journal of economic literature},
  volume={37},
  number={3},
  pages={1120--1149},
  year={1999}
}

@article{libecap1984contractual,
  title={Contractual responses to the common pool: prorationing of crude oil production},
  author={Libecap, Gary D and Wiggins, Steven N},
  journal={The American Economic Review},
  volume={74},
  number={1},
  pages={87--98},
  year={1984},
  publisher={JSTOR}
}

@inproceedings{mcknight2010pay,
  title={Pay me now or pay me more later: start the development of active orbital debris removal now},
  author={McKnight, Darren},
  booktitle={advanced Maui optical and space surveillance technologies conference},
  pages={E63},
  year={2010}
}

@book{world2009sunken,
  title={The sunken billions: the economic justification for fisheries reform},
  author={Kelleher, Kieran and Willmann, Rolf and Arnason, Ragnar and World Bank and FAO},
  year={2009},
  publisher={The World Bank}
}

\newpage

\appendix

\section{Appendix}

\oaEqmExist*
\begin{proof}
Define 
\begin{equation*}
    A_i = \frac{\sum_j p_j (1 - \tau_{ij})(1 + k(Q- D_0))}{kd \sum_j p_j (1 - \tau_{ij}) + m_i}
\end{equation*}
and 
\begin{equation*}
    B_{i} = - \frac{k d \sum_j p_j (1 - \tau_{ij})}{kd \sum_j p_j (1 - \tau_{ij}) + m_i},
\end{equation*} 
where $A_i$ and $B_i$ reflect the long-run benefit-cost ratio and the long-run ratio of environmental to total costs for sector $i$. Equation \ref{eqn:sat_bestresponse} can then be written as 
\begin{equation}
    S_i^\star(S_{-i}, \tau_{i \cdot}, Q) = A_i + B_i S_{-i}.
\end{equation}

Collecting best-response functions in this form yields the linear system
\begin{align}
    S^\star_1 &= A_1 + B_{1}\sum_{j \neq 1} S^\star_j \\
    S^\star_2 &= A_2 + B_{2}\sum_{j \neq 2} S^\star_j \\
    & \vdots \nonumber \\
    S^\star_{n_s} &= A_{n_s} + B_{n_s}\sum_{j \neq n_s} S^\star_j .
\end{align}

Letting $A$ be the $n_s \times 1$ vector of the $A_i$,
\begin{equation}
    A = [A_1, \dots, A_{n_s}],
\end{equation}
$B$ be the $n_s \times n_s$ matrix with $0$ on the main diagonal and row $i$ containing copies of $B_i$, 
\begin{align}
B = 
    \begin{bmatrix}
    0       & B_1   & \dots  & B_1 \\
    B_2     & 0     &        & \vdots \\
    \vdots  &       & \ddots & B_{n_s - 1}\\
    B_{n_s} & \dots & B_{n_s}& 0
    \end{bmatrix}
\end{align}
and $S$ be the $n_s \times 1$ vector of equilibrium satellite levels $S^\star_i$, the linear system above can be rewritten as
\begin{equation}
    S = A + BS.
\end{equation}

If $det(I - B) \neq 0$, then the global open-access equilibrium is
\begin{equation*}
    S = (I - B)^{-1}A.
\end{equation*}
\end{proof}

\oaReduction*
\begin{proof}
Let $I$ be the $n_s \times n_s$ identity matrix, $B$ and $A$ as defined in Lemma \ref{lemm:open_access_eqm} for an $n_s$-player global open-access equilibrium, and $I^2, B^2, A^2$ their analogs for a $2$-player global open-access equilibrium such that $\sum_i A_i = \sum_i A^2_i$ and $\sum_i B_i = \sum_i B^2_i$. Define $G = (I - B)^{-1}$ and $H = (I^2 - B^2)^{-1}$. \\

$G$ can be partitioned into 4 submatrices,
\begin{align}
G = \begin{bmatrix}
G_{11} & G_{12} \nonumber \\
G_{21} & G_{22} \nonumber
\end{bmatrix},
\end{align}
where $G_{11}$ has dimension $1 \times 1$, $G_{12} = G_{21}^T$ by the  has dimension $1 \times (n_s-1)$, and $G_{22}$ has dimension $(n_s-1) \times (n_s-1)$. Similarly, $A$ can be partitioned into 2 vectors,

\begin{align}
A = \begin{bmatrix}
A_{1}\nonumber \\
A_{\setminus 1}\nonumber
\end{bmatrix},
\end{align}

of dimension $1 \times 1$ and $(n_s-1)\times1$. Like $G$ and $A$, $H$ and $A^2$ can be written as 
\begin{align}
H = \begin{bmatrix}
H_{11} & H_{12} \nonumber \\
H_{21} & H_{22} \nonumber
\end{bmatrix}
\end{align}
and
\begin{align}
A^2 = \begin{bmatrix}
A^2_{1}\nonumber \\
A^2_{\setminus 1}\nonumber
\end{bmatrix},
\end{align}

where all elements of both objects are of dimension $1 \times 1$. From the definitions of $A^2$ and $B^2$, it can be seen that
\begin{align}
    G_{11}A_1 + G_{12}A_{-1} &= H_{11}A^2_1 + H_{12}A^2_{-1} \\ 
    \implies S_i^\star &= S_i^{\star,2}.
\end{align}

Similarly, from the definitions of $A^2$ and $B^2$, it can be seen that
\begin{align}
    \sum_{-i} \left( G_{21}A_1 + G_{22}A_{\setminus 1} \right) &= H_{21}A^2_1 + H_{22}A^2_{\setminus 1} \\
    \implies S_{-i}^\star &= S_j^{\star,2}.
\end{align}
where the sum in $\sum_{-i}$ is over the $n_s-1$ elements of the $(n_s-1) \times 1$ vector $G_{21}A_1 + G_{22}A_{-1}$.

\end{proof}

\oaDecomposition*
\begin{proof}
From Lemma \ref{lemm:open_access_eqm} we have that $S = (I - B)^{-1} A$ is the global open-access equilibrium satellite distribution. Define $\phi$ to be the $n_s \times 1$ vector where each element is $1 + kQ$, and $r$ to be the $n_s \times 1$ vector with $i^{\text{th}}$ element
\begin{equation}
    r_i = \frac{\sum_j p_j (1 - \tau_{ij})}{kd \sum_j p_j (1 - \tau_{ij}) + m_i}.
\end{equation}

Next, note that
\begin{align}
    A = \phi \odot r. 
\end{align}

Define
\begin{equation}
    \sigma = (I - B)^{-1} \phi,
\end{equation}

\noindent which is a $n_s \times 1$ vector. Then we have 
\begin{align}
    S &= (I - B)^{-1} A \\
    &= \sigma \odot r,
\end{align}
where only $\sigma$ depends on $Q$.
\end{proof}

\taxesRents*
\begin{proof}
    Differentiating equation \ref{eqn:two_nation_eqm} yields
    \begin{equation}
        \pdv{S_i^\star}{\tau_{ij}} = \pdv{\sigma_i}{r_i}\pdv{r_i}{\tau_{ij}} r_i + \sigma_i \pdv{r_i}{\tau_{ij}}.
    \end{equation}
    Since we have assumed $r_i < (kd)^{-1} ~~ \forall i$, we have $\sigma_i, r_i > 0 ~~ \forall i$. Differentiating $r_i$ and $\sigma_i$ under these assumptions yields
    \begin{align}
        \pdv{\sigma_i}{r_i} &= \frac{(kd)^2 (1 + k(Q-D_0)) (1 - kd r_j) r_j}{(1 - (kd)^2 r_i r_j )^2} > 0, \\
        \pdv{r_i}{\tau_{ij}} &= - \frac{m_i p_j}{(\sum_j p_j (1 - \tau_{ij}) + m_i)^2 } < 0.
    \end{align}
    This establishes that $\forall i,j, ~~ \pdv{S_i^\star}{\tau_{ij}} < 0$. \\
    
    Again differentiating equation \ref{eqn:two_nation_eqm} we obtain
    \begin{equation}
        \pdv{S_j^\star}{\tau_{ij}} = \pdv{\sigma_j}{r_i}\pdv{r_i}{\tau_{ij}} r_j.
    \end{equation}
    Differentiating $\sigma_i$, % and $r_j$,
    \begin{align}
        \pdv{\sigma_j}{r_i} &= - \frac{kd (1 + k(Q-D_0)) (1 - kd r_j)}{(1 - (kd)^2 r_i r_j )^2} < 0, %\\
        % \pdv{r_j}{\tau_{ji}} &= - \frac{m_j p_i}{(\sum_j p_j (1 - \tau_{ij}) + m_j)^2 } < 0,
    \end{align}
    establishing that $\forall i \neq j, ~~ \pdv{S_j^\star}{\tau_{ij}} > 0$. \\
    
    Finally, comparing $-\pdv{S_i^\star}{\tau_{ij}}$ and $\pdv{S_j^\star}{\tau_{ij}}$, we obtain
    \begin{align}
        -\pdv{S_i^\star}{\tau_{ij}} &> \pdv{S_j^\star}{\tau_{ij}}\\
        \implies -\pdv{\sigma_i}{r_i}\pdv{r_i}{\tau_{ij}} r_i - \sigma_i \pdv{r_i}{\tau_{ij}} &>  \pdv{\sigma_j}{r_i}\pdv{r_i}{\tau_{ij}} r_j \\
        \implies -\pdv{r_i}{\tau_{ij}} \left(\pdv{\sigma_i}{r_i}r_i + \pdv{\sigma_j}{r_i}r_j  + \sigma_i \right) &> 0,
    \end{align}
    which holds as long as $\left(\pdv{\sigma_i}{r_i}r_i + \pdv{\sigma_j}{r_i}r_j  + \sigma_i \right)>0$. Simplifying reveals that
    \begin{align}
        \left(\pdv{\sigma_i}{r_i}r_i + \pdv{\sigma_j}{r_i}r_j  + \sigma_i \right) > 0 \iff (1 + k(Q-D_0))[1 - kdr_j(1 - kdr_i(1 -kdr_i))] > 0,
    \end{align}
    which is ensured by $r_i < (kd)^{-1} ~~ \forall i$.
\end{proof}

\abatementGrowth*
\begin{proof}
Differentiating equation \ref{eqn:open_access} yields
\begin{equation}
    \pdv{S_i^\star}{Q} = \frac{k(1 - kd r_j)}{1 - (kd)^2 r_i r_j} > 0,
\end{equation}
establishing that $\forall i, ~~ \pdv{S_i^\star}{Q} > 0$. Differentiating again yields
\begin{equation}
    \frac{\partial^2 S_i^\star}{\partial Q \partial \tau_{ij}} = -\frac{d^2k^3 m_i m_j p_j \sum_i p_i (1 - \tau_{ji}) }{(kdm_j \sum_j p_j (1 - \tau_{ij}) + m_i (kd \sum_i p_i (1 - \tau_{ji}) + m_j))^2 } < 0,
\end{equation}
establishing that $\forall i,j, ~~ \frac{\partial^2 S_i^\star}{\partial Q \partial \tau_{ij}} < 0$.
\end{proof}

\treatyEffective*
\begin{proof}
Differentiating $D^*$ with respect to $Q$, we obtain
\begin{align}
    \pdv{D^*}{Q} &= d \left[ \pdv{S^*_i}{Q} + \pdv{S^*_j}{Q} \right] - 1 \\
    &= d \left[ \frac{k(1 - kdr_j) + k(1 - kdr_i)}{1 - (kd)^2 r_i r_j} \right] - 1 \\
    &= -kd \frac{(1-2kd) + (kd)^2 (r_i + r_j - r_ir_j)}{1 - (kd)^2 r_i r_j}.
\end{align}
First, $m_i > 0 ~~ \forall i \implies r_i < 1 ~~ \forall i \implies r_i + r_j - r_ir_j > 0$. Next, $r_i < (kd)^{-1} ~~ \forall i \implies 1 - (kd)^2 r_i r_j > 0$. Finally, $kd < \frac{1}{2} \implies 1-2kd > 0$. These three conditions are sufficient to ensure $\pdv{D^*}{Q} < 0$.
\end{proof}

\eqmTaxDistn*
\begin{proof}

First, note that the system of bilateral tax best responses defined by equation \ref{eqn:tax_foc} is a bijective mapping from a compact set to itself, i.e.
\begin{equation}
\label{eqn:tax_BR_mapping}
\tau^\star_{\cdot \cdot}(\tau_{\cdot \cdot}) : \mathbb{R}^{n_s \times n_s}_{[0,1]} \to \mathbb{R}^{n_s \times n_s}_{[0,1]}.
\end{equation}
By the Brouwer Fixed Point Theorem, an equilibrium tax distribution exists. It remains to be shown that an equilibrium exists with positive tax rates. \\
Select an arbitrary nation $j$. For conciseness, define 
\begin{equation}
    V_j = p_j \sum_{i=1}^{n_s} (1 - \tau_{ij})  S_i^\star,
\end{equation}
so that $j$'s benefits function can be written as $W_j = (1-kD) V_j$. Next, note that
\begin{equation}
    \pdv{W_j}{\tau_{ij}} = -k \pdv{D}{\tau_{ij}} V_j + (1-kD)\pdv{V_j}{\tau_{ij}},
\end{equation}
where
\begin{align}
    \pdv{D}{\tau_{ij}}\bigg|_{\tau_{ij}=\tau_{jj}=0} &= d \left( \pdv{S^\star_i}{\tau_{ij}} + \pdv{S^\star_j}{\tau_{ij}}
    %+ \pdv{\bar{Q}}{\tau_{ij}}
    \right)\bigg|_{\tau_{ij}=\tau_{jj}=0} < 0, \\
    \pdv{V_j}{\tau_{ij}}\bigg|_{\tau_{ij}=\tau_{jj}=0} &= p_j \left( -S^\star_i +  \left(\pdv{S^\star_i}{\tau_{ij}} + \pdv{S^\star_j}{\tau_{ij}}\right) \right)\bigg|_{\tau_{ij}=\tau_{jj}=0} < 0
\end{align}
both inequalities following from Proposition \ref{prop:taxes_rents}. Rearranging and applying Assumption \ref{assn:tax_improves_welfare}, we obtain 
\begin{equation}
\label{eqn:W_j_tau_ij_positive}
    \pdv{W_j}{\tau_{ij}}\bigg|_{\tau_{ij}=\tau_{jj}=0} > 0.
\end{equation}

Similarly, we compute 
\begin{equation}
    \pdv{W_j}{\tau_{jj}} = -k \pdv{D}{\tau_{jj}} V_j + (1-kD)\pdv{V_j}{\tau_{jj}},
\end{equation}
obtaining
\begin{align}
    \pdv{D}{\tau_{jj}}\bigg|_{\tau_{ij}=\tau_{jj}=0} &= d \left( \pdv{S^\star_i}{\tau_{jj}} + \pdv{S^\star_j}{\tau_{jj}}
    \right)\bigg|_{\tau_{ij}=\tau_{jj}=0} < 0, \\
    \pdv{V_j}{\tau_{ij}}\bigg|_{\tau_{ij}=\tau_{jj}=0} &= p_j \left( -S^\star_i +  \left(\pdv{S^\star_i}{\tau_{ij}} + \pdv{S^\star_j}{\tau_{ij}}\right) \right)\bigg|_{\tau_{ij}=\tau_{jj}=0} < 0
\end{align}

Analysis reveals both inequalities hold whenever $\frac{m_j}{kd \sum_j p_j(1 - \tau_{ij}) } > 0$, which is always the case. Again applying Assumption \ref{assn:tax_improves_welfare}, we obtain 
\begin{equation}
\label{eqn:W_j_tau_jj_positive}
    \pdv{W_j}{\tau_{jj}}\bigg|_{\tau_{ij}=\tau_{jj}=0} > 0.
\end{equation}

Conditions \ref{eqn:W_j_tau_ij_positive} and \ref{eqn:W_j_tau_jj_positive} show that any nation $j$ will prefer to levy a positive tax rate on any nation $i$, including itself. Since existence of an equilibrium is guaranteed by the Brouwer Fixed Point Theorem and we have shown that all nations will prefer positive tax rates to zero tax rates (regardless of what tax rates others levy), we have established existence of a global regulatory equilibrium with positive tax rates.
\end{proof}

\nashAbatement*
\begin{proof}
There are two cases: either it is not privately optimal to avert catastrophe, or it is privately optimal to avert catastrophe. \\

\textbf{If is not privately optimal to avert catastrophe:} Nations choosing abatement levels non-cooperatively will maximize equation \ref{eqn:abatement_payoff_component} taking all other nations' abatement levels as given. This gives the following FOC:
\begin{align}
    q_i : \frac{\partial \pi_i}{\partial q_i} &= - \beta_i - c q_i = 0 \\
    \implies q_i &= -\frac{\beta_i}{c}.
\end{align}

Since $\beta_i > 0$ by Assumption \ref{assn:no_natural_orbital_seizure}, $-\frac{\beta_i}{c} < 0$. But since $q_i$ is bounded below at $0$, the privately optimal abatement level is $q_i=0$. Since this holds for all $i$, we obtain $Q = 0$ and catastrophe is not averted. \\

\textbf{If it is privately optimal to avert catastrophe:} Since the payoff function is quadratic except for a discontinuity at $Q = \bar{Q}$, the maximum is attained either when $q_i = 0$ or when $q_i : Q = \bar{Q}$. \\

Suppose all $N-1$ nations $j \neq i$ play $q_j = \frac{\bar{Q}}{N}$. Then $i$ is pivotal in averting catastrophe. If $i$ chooses not to play $q_i = \frac{\bar{Q}}{N}$, then following the argument in the previous case they will free-ride and play $q_i = 0$, giving the following payoff function:
\begin{align}
    \pi_i = 
    \left\{
    	\begin{array}{ll}
    		b_i(\bar{Q}) - \frac{c(\bar{Q}/N)^2}{2}  & \mbox{if } Q = \bar{Q} \\
    		b_i \left(\frac{\bar{Q}}{N}(N-1) \right) - X & \mbox{if } Q < \bar{Q}
    	\end{array}
    \right. 
\end{align}

Nation $i$ will not free-ride and instead play $q_i = \frac{\bar{Q}}{N}$ if and only if
\begin{align}
    b_i(\bar{Q}) - \frac{c(\bar{Q}/N)^2}{2} &\geq b_i \left(\frac{\bar{Q}}{N}(N-1) \right) - X \\
    \implies X &\geq b_i \left(\frac{\bar{Q}}{N}(N-1) \right) - b_i(\bar{Q}) + \frac{c(\bar{Q}/N)^2}{2} \\
    \implies X &\geq \beta_i \frac{\bar{Q}}{N} + \frac{c(\bar{Q}/N)^2}{2}.
    \tag{\ref{eqn:no_treaty_defection}}
\end{align}

Equation \ref{eqn:no_treaty_defection} must hold for all $i$ to sustain the catastrophe-averting Nash equilibrium without a treaty.
\end{proof}

\selfEnforcingTreaty*
\begin{proof}
Starting from full participation in the treaty, suppose nation $i$ considers withdrawal. If they do so, the treaty instructs the remaining $N-1$ members to play $Q_{\setminus i}$ (which may be different from $\frac{\bar{Q}}{N}(N-1)$). There are two cases from here:
\begin{enumerate}
    \item either $i$ plays $q_i = \bar{Q} - Q_{\setminus i}$ and catastrophe is averted despite their defection, or
    \item $i$ plays $q_i = 0$ and allows catastrophe to occur.
\end{enumerate}
Suppose $i$ considers withdrawal and supplying $q_i=0$ units of abatement instead of the treaty-specified $\frac{\bar{Q}}{N}$, and the remaining members play $Q_{\setminus i}$. Then $i$ will prefer to avert catastrophe by supplying $q_i = \bar{Q} - Q_{\setminus i}$ units of abatement instead of $q_i = 0$ if and only if
\begin{equation}
        \underbrace{b_i(\bar{Q}) - \frac{c}{2}\left(\bar{Q}-Q_{\setminus i}\right)^2}_{\text{\shortstack{Payoff from averting catastrophe\\while outside the treaty}}}   \geq \underbrace{b_i\left(Q_{\setminus i}\right) - X}_{\text{\shortstack{Payoff from Nash defection\\and incurring catastrophe}}}\label{eqn:leave_treaty_avert_catastrophe}
\end{equation}

Since $Q_{\setminus i}$ is a choice for the remaining treaty signatories, condition \ref{eqn:leave_treaty_avert_catastrophe} shows how the treaty can be written to naturally induce any defector to avert catastrophe despite leaving the treaty. Rearranging condition \ref{eqn:leave_treaty_avert_catastrophe}, solving the quadratic inequality, and selecting the smallest root, we obtain
\begin{equation}
\tag{\ref{eqn:treaty_response_abatement}}
    Q_{\setminus i} = \bar{Q} + \frac{1}{c} \left( \beta_i - \sqrt{\beta_i^2 + 2cX} \right).
\end{equation}
Since $\beta_i, c, X > 0$, we have $\left( \beta_i - \sqrt{\beta_i^2 + 2cX} \right) < 0$, and $Q_{\setminus i} < \bar{Q}$. \\

Comparing the treaty response abatement quantity $Q_{\setminus i}$ to the treaty abatement quantity $\frac{\bar{Q}}{N}$ shows that if
\begin{equation}
\tag{\ref{eqn:defector_gets_punished}}
    \frac{\bar{Q}}{N} < - \frac{1}{c} \left( \beta_i - \sqrt{\beta_i^2 + 2cX} \right),
\end{equation}
then the treaty response $Q_{\setminus i}$ coordinates the remaining signatories to ``punish'' defection by nation $i$ by inducing them to supply abatement greater than $\frac{\bar{Q}}{N}$. Else, the treaty response coordinates the remaining signatories to avert catastrophe by reducing $i$'s burden of cooperation. Thus, if condition \ref{eqn:defector_gets_punished} holds, defections by nation $i$ are costlier than remaining in the treaty and they will prefer to remain in the treaty.
\end{proof}

\betaiCompStatics*
\begin{proof}
For conciseness, define
\begin{align}
    A_{i} &\equiv \sum_j p_j (1-\tau_{ij}), \\
    B_{i} &\equiv kd A_i + m_i, \\
    T^i_{ij} &\equiv \frac{2k^3d m_i m_j^2 A_i^2 \left( 3m_i A_j + kd A_i \left( 3m_j + kd A_j \right) \right) }{B_i^3 B_j \left(kd m_jA_i + m_i B_j \right)^2 } \\
    T^j_{ij} &\equiv \frac{2k^4d^2m_j^2 A_i^3 \left( kd m_j A_i + 2m_i B_j \right)}{\left( A_i A_j \left( kd m_j A_i + m_i B_j \right) \right)^2 }
\end{align}
Next, computing derivatives, we obtain:
\begin{align}
    \pdv{\beta_i}{\tau_{ii}} &= - p_i T^i_{ij} < 0 \\
    \pdv{\beta_i}{\tau_{ij}} &= - p_j T^i_{ij} < 0 \\
    \pdv{\beta_i}{\tau_{ji}} &= - p_i T^j_{ij} < 0 \\
    \pdv{\beta_i}{\tau_{jj}} &= - p_j T^j_{ij} < 0 \\
    \pdv{\beta_i}{m_i} &= - \frac{2k^3dm_j^2 A_i^3 \left( 3m_i B_j + kd A_i \left( 3m_j + kd A_j \right) \right)}{ B_i^3 B_j \left( kd m_j A_i + m_i B_j \right)^2 } < 0.
\end{align}
\end{proof}

\taxesSupportTreaties*
\begin{proof}
We first establish condition \ref{eqn:tax_makes_nash_aversion_better}, which contains the right-hand side of condition \ref{eqn:no_treaty_defection}:
\begin{align}
    \pdv{}{\tau_{ij}} \left( \beta_i \frac{\bar{Q}}{N} + \frac{c(\bar{Q}/N)^2}{2} \right) &= \pdv{\beta_i}{\tau_{ij}} \frac{\bar{Q}}{N} + \frac{\beta_i}{N} \pdv{\bar{Q}}{\tau_{ij}} + \frac{c \bar{Q}}{N^2} \pdv{\bar{Q}}{\tau_{ij}} \\
    &= \pdv{\beta_i}{\tau_{ij}} \frac{\bar{Q}}{N} + \frac{\beta_i}{N} d \pdv{S^\star}{\tau_{ij}} + \frac{c \bar{Q}}{N^2} \pdv{S^\star}{\tau_{ij}} < 0.
\end{align}
The inequality follows from Proposition \ref{prop:taxes_rents}, which establishes that $\pdv{S^\star}{\tau_{ij}}$, and Lemma \ref{lemm:beta_i_statics}, which establishes that $\pdv{\beta_i}{\tau_{ij}}<0$. Since increases in $\tau{ij}$ reduce the right-hand side of condition \ref{eqn:no_treaty_defection} (making the condition more likely to hold), they increase $i$'s incentive to avert Kessler Syndrome. \\

Next, we establish condition \ref{eqn:tax_makes_defection_worse}, which contains both the sides of condition \ref{eqn:defector_gets_punished}:
\begin{align}
    \pdv{}{\tau_{ij}} \left(- \frac{\bar{Q}}{N} - \frac{1}{c} \left( \beta_i - \sqrt{\beta_i^2 + 2cX} \right) \right) &= \frac{1}{N}\pdv{\bar{Q}}{\tau_{ij}} - \pdv{\beta_i}{\tau_{ij}} \left( \frac{1}{c} - \frac{\beta_i}{\sqrt{\beta_i^2 + 2 c X}} \right).
\end{align}
While the sign of  $\frac{1}{c} - \frac{\beta_i}{\sqrt{\beta_i^2 + 2 c X}}$ is ambiguous, some manipulation reveals that $\sqrt{\beta_i^2 + 2cX} - c \beta_i > 0$ is a sufficient condition to ensure $\frac{1}{c} - \frac{\beta_i}{\sqrt{\beta_i^2 + 2 c X}} > 0$. Thus, again applying Proposition \ref{prop:taxes_rents} and Lemma \ref{lemm:beta_i_statics}, we obtain the desired result:
\begin{align}
    \pdv{}{\tau_{ij}} \left(- \frac{\bar{Q}}{N} - \frac{1}{c} \left( \beta_i - \sqrt{\beta_i^2 + 2cX} \right) \right) &= \frac{1}{N}\pdv{\bar{Q}}{\tau_{ij}} - \pdv{\beta_i}{\tau_{ij}} \left( \frac{1}{c} - \frac{\beta_i}{\sqrt{\beta_i^2 + 2 c X}} \right) > 0.
\end{align}
\end{proof}

\end{document}